\documentclass{amsart}
\usepackage{amsmath,amssymb,amsthm}
\newcommand{\defeq}{\; {\overset{\text{def}}=} \;}
\DeclareMathOperator{\Ad}{Ad}
\DeclareMathOperator{\ad}{ad}

\newcommand{\calB}{{\mathcal B}}
\newcommand{\calF}{{\mathcal F}}
\newcommand{\calL}{{\mathcal L}}
\newcommand{\calM}{{\mathcal M}}
\newcommand{\calP}{{\mathcal P}}
\newcommand{\calQ}{{\mathcal Q}}
\newcommand{\calX}{{\mathcal X}}
\newcommand{\commentout}[1]{}
\newcommand{\der}{\partial}

\DeclareMathOperator{\ord}{ord}
\DeclareMathOperator{\ordh}{ord^{\hbar}}
\DeclareMathOperator{\ordx}{ord_\xi}
\DeclareMathOperator{\Res}{Res}
\newcommand{\symh}{\sigma^\hbar}
\newcommand{\totsym}{\sigma_{\mathrm{tot}}}
\newtheorem{thm}{Theorem}[section]
\newtheorem{lem}[thm]{Lemma}
\newtheorem{prop}[thm]{Proposition}

\theoremstyle{definition}

\theoremstyle{remark}

\newtheorem{rem}[thm]{Remark}

\numberwithin{equation}{section}

\newcommand\thmref[1]{Theorem~\ref{#1}}
\newcommand\propref[1]{Proposition~\ref{#1}}
\newcommand\secref[1]{Section~\ref{#1}}
\newcommand\appref[1]{Appendix~\ref{#1}}

\newcommand\lemref[1]{Lemma~\ref{#1}}

\begin{document}
\title[$\hbar$-expansion of KP]{
       $\hbar$-expansion of KP hierarchy:\\
       recursive construction of solutions
}
\author{Kanehisa Takasaki}
\address{
Graduate School of Human and Environmental Studies,
Kyoto University,
Yoshida, Sakyo, Kyoto, 606-8501, Japan
}
\email{takasaki@math.h.kyoto-u.ac.jp}
\author{Takashi Takebe}
\address{
Faculty of Mathematics,
State University -- Higher School of Economics,
Vavilova Street, 7, Moscow, 117312, Russia
}
\email{ttakebe@hse.ru}
\date{24 December 2009; revised on 5 September 2010}
%
%
\maketitle
\begin{abstract}
The $\hbar$-dependent KP hierarchy is a formulation of 
the KP hierarchy that depends on the Planck constant $\hbar$ 
and reduces to the dispersionless KP hierarchy as $\hbar \to 0$.  
A recursive construction of its solutions on the basis of 
a Riemann-Hilbert problem for the pair $(L,M)$ of 
Lax and Orlov-Schulman operators
is presented.  
The Riemann-Hilbert problem is converted to a set 
of recursion relations for the coefficients $X_n$ 
of an $\hbar$-expansion of the operator 
$X = X_0 + \hbar X_1 + \hbar^2X_2 + \cdots$ 
for which the dressing operator $W$ is expressed 
in the exponential form $W = \exp(X/\hbar)$.  
Given the lowest order term $X_0$, 
one can solve 
the recursion relations to obtain the higher order terms. 
The wave function $\Psi$ associated with $W$ turns out 
to have the WKB form $\Psi = \exp(S/\hbar)$, 
and the coefficients $S_n$ of the $\hbar$-expansion 
$S = S_0 + \hbar S_1 + \hbar^2 S_2 + \cdots$, 
too, are determined by a set of recursion relations. 
This WKB form is used to show that the associated 
tau function has an $\hbar$-expansion of the form 
$\log\tau = \hbar^{-2}F_0 + \hbar^{-1}F_1 + F_2 + \cdots$. 

\end{abstract}

\setcounter{section}{-1}
\section{Introduction}
\label{sec:intro}

The KP hierarchy can be completely solved by 
several methods. The most classical methods are 
based on Grassmann manifolds \cite{sat-sat:82}, \cite{seg-wil:85},
fermions and vertex operators \cite{djkm} and 
factorisation of microdifferential operators \cite{mul:84}.
Unfortunately, those methods are not very suited for 
a ``quasi-classical'' ($\hbar$-dependent, where $\hbar$ 
is the Planck constant) formulation \cite{tak-tak:95} 
of the KP hierarchy.  

The $\hbar$-dependent formulation of the KP hierarchy 
was introduced to study the dispersionless KP hierarchy 
\cite{kod-gib}, \cite{kri:91}, \cite{tak-tak:91} as a classical limit 
(i.e., the lowest order of the $\hbar$-expansion) 
of the KP hierarchy.   This point of view turned out 
to be very useful for understanding various features 
of the dispersionless KP hierarchy such as 
Lax equations, Hirota equations, infinite dimensional 
symmetries, etc., in the light of the KP hierarchy.  
In this paper, we return to the $\hbar$-dependent 
KP hierarchy itself, and consider all orders 
of the $\hbar$-expansion.  

We first address the issue of solving a Riemann-Hilbert problem for the
pair $(L,M)$ of Lax and Orlov-Schulman operators \cite{orl-sch}. This
is a kind of ``quantisation'' of a Riemann-Hilbert problem that solves
the dispersionless KP hierarchy \cite{tak-tak:91}.  Though the
Riemann-Hilbert problem for the full KP hierarchy was formulated in our
previous work \cite{tak-tak:95}, we did not consider the existence of
its solution in a general setting.  In this paper, we settle this issue
by an $\hbar$-expansion of the dressing operator $W$, which is assumed
to have the exponential form $W = \exp(X/\hbar)$ with an operator $X$ of
negative order.  Roughly speaking, the coefficients $X_n$, $n =
0,1,2,\ldots$, of the $\hbar$-expansion of $X$ are shown to be
determined recursively from the lowest order term $X_0$ (in other words,
from a solution of the dispersionless KP hierarchy). 

We next convert this result to the language of 
the wave function $\Psi$.  This, too, is to answer a problem 
overlooked in our previous paper \cite{tak-tak:95}.   
Namely, given the dressing operator 
in the exponential form $W = \exp (X/\hbar)$, 
we show that the associated wave function has the WKB form 
$\Psi=\exp(S/\hbar)$ with a phase function $S$ expanded 
into nonnegative powers of $\hbar$.   This is genuinely 
a problem of calculus of microdifferential operators 
rather than that of the KP hierarchy.  A simplest example 
such as $X=x(\hbar\der)^{-1}$ demonstrates that this problem 
is by no means trivial.  Borrowing an idea from Aoki's 
``exponential calculus'' of microdifferential operators \cite{aok:86}, 
we show that dressing operators of the form $W = \exp(X/\hbar)$ 
and wave functions of the form $\Psi = \exp(S/\hbar)$ 
are determined from each other by a set of recursion relations 
for the coefficients of their $\hbar$-expansion. 
More precisely, we need many auxiliary quantities other than $X$ and
$S$, for which we can derive a large set of recursion relations. Thus
our construction is essentially recursive.
Consequently, 
the wave function of the solution of the aforementioned 
Riemann-Hilbert problem, too, are recursively determined 
by the $\hbar$-expansion.  

Having the $\hbar$-expansion of the wave function, 
we can readily derive an $\hbar$-expansion of the tau function 
as stated in our previous work \cite{tak-tak:95}.  
This $\hbar$-expansion is a generalisation of 
the ``genus expansion'' of partition functions in string theories 
and random matrices \cite{dij:91}, \cite{kri:91}, \cite{mor:94},
\cite{dfgz}.

This paper is organised as follows.  
Section 1 is a review of the $\hbar$-dependent formulation 
of the KP hierarchy.  Relevant Riemann-Hilbert problems 
are also reviewed here.  
Section 2 presents the recursive solution of 
the Riemann-Hilbert problem.  A technical clue is 
the Campbell-Hausdorff formula, details of which 
are collected in Appendix A.  The construction 
of solution is illustrated for the case of 
the Kontsevich model \cite{adl-vaM:92} in Appendix B. 
Section 3 deals with the $\hbar$-expansion of the wave function. 
Aoki's exponential calculus is also briefly reviewed here.  
Section 4 mentions the $\hbar$-expansion of the tau function. 
Section 5 is devoted to concluding remarks. 

\bigskip

\paragraph*{\em Acknowledgements}

The authors are grateful to Professor Akihiro Tsuchiya 
for drawing our attention to this subject and to 
Professor Masatoshi Noumi for instructing how to use 
Mathematica in algebra of microdifferential operators.  
Computations in Appendix B were done with the aid of 
one of his Mathematica programmes.  
This work is partly supported by Grants-in-Aid for 
Scientific Research No.\ 19540179 and No.\ 22540186 from the Japan
Society for the Promotion of Science.  
TT is partly supported by the grant of the State University -- Higher
School of Economics, Russia, for the Individual Research Project
09-01-0047 (2009).

\section{$\hbar$-dependent KP hierarchy: review}
\label{sec:dkp}

In this section we recall several facts on the KP hierarchy depending on
a formal parameter $\hbar$ in \cite{tak-tak:95}, \S1.7.

The $\hbar$-dependent KP hierarchy is defined by the Lax representation
\begin{equation}
     \hbar \frac{\der L}{\der t_n} = [ B_n, L ], \qquad
    B_n = (L^n)_{\geq 0}, \qquad n=1,2,\ldots,
\label{kph}
\end{equation}
where the {\em Lax operator} $L$ is a microdifferential operator of the
form
\begin{equation}
    L = \hbar \der
      + \sum_{n=1}^\infty u_{n+1}(\hbar,x,t)(\hbar\der)^{-n}, \qquad
    \der = \frac{\der}{\der x},
\label{L}
\end{equation}
and ``$(\quad)_{\geq 0}$'' stands for the projection onto a differential
operator dropping negative powers of $\der$. The coefficients
$u_n(\hbar,x,t)$ of $L$ are assumed to be formally regular with respect
to $\hbar$. This means that they have an asymptotic expansion of the
form $u_n(\hbar,x,t) = \sum_{m=0}^\infty \hbar^m u_n^{(m)}(x,t)$ as $\hbar
\to 0$.

We need two kinds of ``order'' of microdifferential operators: one is
the ordinary order,
\begin{equation}
    \ord \left( \sum a_{n,m}(x,t) \hbar^n \der^m \right)
    \defeq
    \max \left\{ m 
       \,\left|\, \sum_{n} a_{n,m}(x,t)\hbar^n \neq 0 
         \right.\right\},
\label{def:ord}
\end{equation}
and the other is the $\hbar$-order defined by
\begin{equation}
    \ordh \left( \sum a_{n,m}(x,t) \hbar^n \der^m \right)
    \defeq
    \max \{ m-n \,|\, a_{n,m}(x,t) \neq 0 \}.
\label{def:ordh}
\end{equation}
In particular, $\ordh\hbar=-1$, $\ordh\der=1$, $\ordh\hbar\der=0$. For
example, the condition which we imposed on the coefficients
$u_n(\hbar,x,t)$ can be restated as $\ordh(L) = 0$.

The {\em principal symbol} (resp.\ the {\em symbol of order $l$})
of a microdifferential operator $A = \sum a_{n,m}(x,t)\hbar^n \der^m$ with
respect to the $\hbar$-order is
\begin{align}
    \symh(A)
    &\defeq \sum_{m-n = \ord(A)} a_{n,m}(x,t) \xi^m 
\label{def:pr-symbol}
    \\
    (\text{resp. } 
    \symh_l(A)
    &\defeq \sum_{m-n = l} a_{n,m}(x,t) \xi^m).
\label{def:symbol}
\end{align}
When it is clear from the context, we sometimes use $\symh$ instead
of $\symh_l$.

\begin{rem}
This ``order'' coincides with the order of an microdifferential
operator if we formally replace $\hbar$ with $\der_{t_0}^{-1}$, where
$t_0$ is an extra variable. In fact, naively extending \eqref{kph} to
$n=0$, we can introduce the time variable $t_0$ on which nothing
depends. See also \cite{kas-rou:08}.
\end{rem}

As in the usual KP theory, the Lax operator $L$ is expressed by a {\em
dressing operator} $W$:
\begin{equation}
    L = \Ad W (\hbar\der) = W (\hbar\der) W^{-1}
\label{L=Ad(W)d}
\end{equation}
The dressing operator $W$ should have a specific form:
\begin{gather}
    W = \exp( \hbar^{-1} X(\hbar,x,t,\hbar\der) )
        (\hbar\der)^{\alpha(\hbar)/\hbar},
\label{W=exp(X)}
\\
    X(\hbar,x,t,\hbar\der) 
    = \sum_{k=1}^\infty \chi_k(\hbar,x, t) (\hbar\der)^{-k},
\label{X}
\\
    \ordh (X(\hbar,x, t, \hbar\der)) = \ordh \alpha(\hbar) = 0,
\label{ordX=0}
\end{gather}
and $\alpha(\hbar)$ is a constant with respect to $x$ and $t$. (In
\cite{tak-tak:95} we did not introduce $\alpha$, which will be necessary
in \secref{sec:recursion}.)

The {\em wave function} $\Psi(\hbar,x,t;z)$ is defined by
\begin{equation}
    \Psi(\hbar,x,t;z) 
    = W e^{(xz + \zeta(t,z))/\hbar},
\label{def:wave-func}
\end{equation}
where $\zeta(t,z)=\sum_{n=1}^\infty t_n z^n$. It is a solution of
linear equations
\begin{equation*}
    L \Psi = z \Psi, \qquad 
    \hbar\frac{\der \Psi}{\der t_n} = B_n \Psi \quad (n=1,2,\dotsc),
\end{equation*}
and has the WKB form \eqref{wave-func}, as we shall show in
\secref{sec:wave-function}. Moreover it is expressed by means of the
{\em tau function} $\tau(\hbar,t)$ as follows:
\begin{equation}
    \Psi(\hbar,x,t;z) =
    \frac{\tau(t+x-\hbar[z^{-1}])}{\tau(t)}
    e^{\hbar^{-1}\zeta(t,z)},
\label{tau/tau}
\end{equation}
where $t+x=(t_1+x,t_2,t_3,\dotsc)$ and
$[z^{-1}]=(1/z,1/2z^2,1/3z^3,\dots)$. We shall study the
$\hbar$-expansion of the tau function in \secref{sec:tau-function}.

The {\em Orlov-Schulman operator} $M$ is defined by 
\begin{equation}
    M = 
    \Ad \left(
          W \exp\left(\hbar^{-1} \zeta(t,\hbar\der) \right)
        \right)x 
    =
    W \left( \sum_{n=1}^\infty n t_n (\hbar \der)^{n-1} + x \right)
    W^{-1}
\label{def:M}
\end{equation}
where $\zeta(t,\hbar\der) = \sum_{n=1}^\infty t_n (\hbar\der)^n$. It is easy
to see that $M$ has a form
\begin{equation}
    M = \sum_{n=1}^\infty n t_n L^{n-1} + x + \alpha(\hbar) L^{-1} +
        \sum_{n=1}^\infty v_n(\hbar,x, t) L^{-n-1},
\label{M}
\end{equation}
and satisfies the following properties:
\begin{itemize}
\item $\ordh(M) = 0$;
\item the canonical commutation relation: $[L, M] = \hbar$;
\item the same Lax equations as $L$:
\begin{equation}
    \hbar \frac{\der M}{\der t_n} = [ B_n, M ],\quad
    n=1,2,\ldots.
\label{lax:M}
\end{equation}
\item another linear equation for the wave function $\Psi$:
\begin{equation*}
    M \Psi = \hbar\frac{\der \Psi}{\der z}.
\end{equation*}
\end{itemize}

\begin{rem}
\label{rem:alpha=const}
 If an operator $M$ of the form \eqref{M} satisfies the Lax equations
 \eqref{lax:M} and the canonical commutation relation $[L,M]=\hbar$ with
 the Lax operator $L$ of the KP hierarchy, then $\alpha(\hbar)$ in the
 expansion \eqref{M} does not depend on any $t_n$ nor on $x$. In
 fact, expanding the canonical commutation relation, we have
\[
    \hbar + \hbar\frac{\der\alpha}{\der x} (\hbar\der)^{-1} 
          + (\text{lower order terms}) 
    = \hbar,
\]
 which implies $\frac{\der\alpha}{\der x}=0$. Similarly, from
 \eqref{lax:M} follows $\frac{\der\alpha}{\der t_n}=0$ with the help of
 \eqref{kph} and $[L^n, M] = n\hbar L^{n-1}$.
\end{rem}

The following proposition (Proposition 1.7.11 of \cite{tak-tak:95}) is a
``dispersionful'' counterpart of the theorem for the dispersionless KP
hierarchy found earlier (Proposition 7 of \cite{tak-tak:91}; cf.\
\propref{prop:dRH} below).

\begin{prop}
\label{prop:RH}
(i)
 Suppose that operators $f(\hbar,x,\hbar\der)$, $g(\hbar,x,\hbar\der)$,
 $L$ and $M$ satisfy the following conditions:
\begin{itemize}
 \item $\ordh f = \ordh g = 0$, $[f,g]=\hbar$;
 \item $L$ is of the form \eqref{L} and $M$ is of the form \eqref{M},
       $[L,M]=\hbar$;
 \item $f(\hbar,M,L)$ and $g(\hbar,M,L)$ are differential operators:
\begin{equation}
    (f(\hbar,M,L))_{<0} = (g(\hbar,M,L))_{<0} = 0,
\label{f(M,L)=g(M,L)=diff}
\end{equation}
       where $(\quad)_{<0}$ is the projection to the negative order
       part: $P_{<0}:=P-P_{\geq 0}$.
\end{itemize}
 Then $L$ is a solution of the KP hierarchy \eqref{kph} and $M$ is the
 corresponding Orlov-Schulman operator.

(ii)
 Conversely, for any solution $(L,M)$ of the $\hbar$-dependent KP
 hierarchy there exists a pair $(f,g)$ satisfying the conditions in (i).
\end{prop}

The leading term of this system with respect to the $\hbar$-order gives
the {\em dispersionless KP hierarchy}. Namely,
\begin{equation}
    \calL := \symh(L)
    = \xi
    + \sum_{n=1}^\infty u_{0,n+1} \xi^{-n}, \qquad
    (u_{0,n+1} := \symh(u_{n+1}))
\label{calL}
\end{equation}
satisfies the dispersionless Lax type equations
\begin{equation}
    \frac{\der \calL}{\der t_n} = \{ \calB_n, \calL\}, \qquad
    \calB_n = (\calL^n)_{\geq 0}, \qquad n=1,2,\ldots,
\label{dkp}
\end{equation}
where $(\quad)_{\geq 0}$ is the truncation of Laurent series to its
polynomial part and $\{,\}$ is the Poisson bracket defined by
\begin{equation}
    \{a(x,\xi), b(x,\xi)\}
    =
    \frac{\der a}{\der \xi} \frac{\der b}{\der x}
    -
    \frac{\der a}{\der x} \frac{\der b}{\der \xi}.
\label{poisson}
\end{equation}

The dressing operation \eqref{L=Ad(W)d} for $L$ becomes the following
dressing operation for $\calL$:
\begin{equation}
 \begin{split}
    \calL &= \exp \bigl( \ad_{\{,\}} X_0 \bigr) 
            \exp \bigl( \ad_{\{,\}} \alpha_0 \log\xi \bigr) \xi
           = \exp \bigl( \ad_{\{,\}} X_0 \bigr) \xi,
\\
    X_0 &:= \symh(X), \ \alpha_0:=\symh(\alpha)
\end{split}
\label{calL=exp(adX)xi}
\end{equation}
where $\ad_{\{,\}} (f) (g):= \{f,g\}$. 

The principal symbol of the Orlov-Schulman operator is
\begin{equation}
    \calM = \sum_{n=1}^\infty nt_n \calL^{n-1} + x 
    + \alpha_0 \calL^{-1}
    + \sum_{n=1}^\infty v_{0,n} \calL^{-n-1}, \qquad
    v_{0,n}:=\symh(v_n),
\label{calM}
\end{equation}
which is equal to
\begin{equation}
    \calM = 
    \exp \bigl( \ad_{\{,\}} X_0 \bigr) 
    \exp \bigl( \ad_{\{,\}} \alpha_0 \log\xi \bigr)
    \exp \bigl( \ad_{\{,\}} \zeta(t,\xi)\bigr) x,
\label{calM=exp(adX)x}
\end{equation}
where $\zeta(t,\xi) = \sum_{n=1}^\infty t_n \xi^n$. The series $\calM$
satisfies the canonical commutation relation with $\calL$,
$\{\calL,\calM\}=1$ and the Lax type equations:
\begin{equation}
    \frac{\der\calM}{\der t_n} = \{\calB_n, \calM\}, \qquad
    n=1,2,\ldots.
\label{lax:calM}
\end{equation}

The Riemann-Hilbert type construction of the solution is essentially the
same as \propref{prop:RH}. (We do not need to assume the canonical
commutation relation $\{\calL,\calM\}=1$.)

\begin{prop}
\label{prop:dRH}
(i)
 Suppose that functions $f_0(x,\xi)$, $g_0(x,\xi)$, $\calL$ and $\calM$
 satisfy the following conditions:
\begin{itemize}
 \item $\{f_0,g_0\}=1$;
 \item $\calL$ is of the form \eqref{calL} and $\calM$ is of the form
       \eqref{calM}
 \item $f_0(\calM,\calL)$ and $g_0(\calM,\calL)$ do not contain negative
       powers of $\xi$
\begin{equation}
    (f_0(\calM,\calL))_{<0} = (g_0(\calM,\calL))_{<0} = 0,
\label{f0(M,L)=g0(M,L)=pos}
\end{equation}
       where $(\quad)_{<0}$ is the projection to the negative degree part:
       $P_{<0}:=P-P_{\geq 0}$.
\end{itemize}
 Then $\calL$ is a solution of the dispersionless KP hierarchy
 \eqref{dkp} and $\calM$ is the corresponding Orlov-Schulman function.

(ii)
 Conversely, for any solution $(\calL,\calM)$ of the dispersionless KP
 hierarchy, there exists a pair $(f_0,g_0)$ satisfying the conditions in
 (i). 
\end{prop}

If $f$, $g$, $L$ and $M$ are as in \propref{prop:RH}, then
$f_0=\symh(f)$, $g_0=\symh(g)$, $\calL=\symh(L)$ and
$\calM=\symh(M)$ satisfy the conditions in \propref{prop:dRH}. In other
words, $(f,g)$ and $(L,M)$ are quantisation of the canonical
transformations $(f_0,g_0)$ and $(\calL,\calM)$ respectively. (See, for
example, \cite{sch:85} for quantised canonical transformations.)

\section{Recursive construction of the dressing operator}
\label{sec:recursion}

In this section we prove that the solution of the KP hierarchy
corresponding to the quantised canonical transformation $(f,g)$ is
recursively constructed from its leading term, i.e., the solution of the
dispersionless KP hierarchy corresponding to the Riemann-Hilbert data
$(\symh(f),\symh(g))$.

Given the pair $(f,g)$, we have to construct the dressing operator $W$,
or $X$ and $\alpha$ in \eqref{W=exp(X)}, such that operators
\begin{equation}
 \begin{aligned}
    f(\hbar,M,L) &= 
    \Ad \left(
          W \exp\left(\hbar^{-1} \zeta(t,\hbar\der) \right)
        \right) f(\hbar,x,\hbar\der)
\\
    g(\hbar,M,L) &= 
    \Ad \left(
          W \exp\left(\hbar^{-1} \zeta(t,\hbar\der) \right)
        \right) g(\hbar,x,\hbar\der)
 \end{aligned}
\label{Ad(W)(f,g)}
\end{equation}
are both differential operators (cf.\ \propref{prop:RH}). Let us expand
$X$ and $\alpha$ with respect to the $\hbar$-order as follows:
\begin{align}
    X(\hbar,x,t,\hbar\der) 
    &= \sum_{n=0}^\infty \hbar^n X_n(x,t,\hbar\der),\quad
    X_n(x,t,\hbar\der) =
    \sum_{k=1}^\infty \chi_{n,k}(x,t) (\hbar\der)^{-k},
\label{X:h-expansion}
\\
    \alpha(\hbar) &= \sum_{n=0}^\infty \hbar^n \alpha_n,
\label{alpha:h-expansion}
\end{align}
where $\chi_{n,k}$ and $\alpha_n$ do not depend on $\hbar$ and hence
$\chi_k$ in \eqref{X} is expanded as
$\chi_k=\sum_{n=0}^\infty\hbar^n\chi_{n,k}$.

Assume that the solution of the dispersionless KP hierarchy
corresponding to $(\symh(f),\symh(g))$ is given. In other words,
assume that a symbol $X_0=\sum_{k=1}^\infty \chi_{0,k}(x,t) \xi^{-k}$
and a constant $\alpha_0$ are given such that
\begin{equation*}
 \begin{aligned}
    \symh(f)(\calL,\calM) &=
    \exp \bigl( \ad_{\{,\}} X_0 \bigr)
    \exp \bigl( \ad_{\{,\}} \alpha_0 \log \xi \bigr)
    \exp \bigl( \ad_{\{,\}} \zeta(t,\xi) \bigr)
    \symh(f)(x,\xi)
\\
    \symh(g)(\calL,\calM) &= 
    \exp \bigl( \ad_{\{,\}} X_0 \bigr)
    \exp \bigl( \ad_{\{,\}} \alpha_0 \log \xi \bigr)
    \exp \bigl( \ad_{\{,\}} \zeta(t,\xi) \bigr)
    \symh(g)(x,\xi)
 \end{aligned}
\end{equation*}
do not contain negative powers of $\xi$:
\begin{equation}
    \bigl(\symh(f)(\calL,\calM)\bigr)_{<0}=
    \bigl(\symh(g)(\calL,\calM)\bigr)_{<0}=0.
\label{sym(f,g)(L,M)-=0}
\end{equation}
(See \propref{prop:dRH}.)

We are to construct $X_n$ and $\alpha_n$ recursively, starting from
$X_0$ and $\alpha_0$. For this purpose expand $f(\hbar,L,M)$ and
$g(\hbar,L,M)$ in \eqref{Ad(W)(f,g)} as follows:
\begin{align}
    P&:= 
    \Ad \left( \exp (\hbar^{-1} X) 
               \exp(\hbar^{-1}\alpha\log\hbar\der) \right)
    f_t
\label{P}
\\
    &=
    P_0 + \hbar P_1 + \cdots + \hbar^k P_k + \cdots,
\nonumber
\\
    Q&:= 
    \Ad \left( \exp (\hbar^{-1} X) 
               \exp(\hbar^{-1}\alpha\log\hbar\der) \right)
    g_t
\label{Q}
\\
    &=
    Q_0 + \hbar Q_1 + \cdots + \hbar^k Q_k + \cdots,
\nonumber
\end{align}
where 
\begin{equation}
    f_t :=
    \Ad \left( e^{\hbar^{-1} \zeta(t,\hbar\der)} \right) f, \qquad
    g_t :=
    \Ad \left( e^{\hbar^{-1} \zeta(t,\hbar\der)} \right) g,
\label{ft,gt}
\end{equation}
and $P_i$'s and $Q_i$'s are operators of the form
$P_i=P_i(x,t,\hbar\der)$ and $Q_i=Q_i(x,t,\hbar\der)$ and $\ordh
P_i=\ordh Q_i =0$. Suppose that we have chosen $X_0,\dots,X_{i-1}$ and
$\alpha_0,\dots,\alpha_{i-1}$ so that $P_0,\dots,P_{i-1}$ and
$Q_0,\dots,Q_{i-1}$ do not contain negative powers of $\der$. If an
operator $X_i$ and a constant $\alpha_i$ are constructed from these
given $X_0,\dots,X_{i-1}$ and $\alpha_0,\dots,\alpha_{i-1}$ so that
resulting $P_i$ and $Q_i$ do not contain negative powers of $\der$, this
procedure gives recursive construction of $X$ and $\alpha$ in question.

We can construct such $X_i$ and $\alpha_i$ as follows. (Details and
meaning shall be explained in the proof of \thmref{thm:recursion:X}.):
\begin{itemize}
 \item (Step 0) Assume $X_0,\dots,X_{i-1}$ and
       $\alpha_0,\dots,\alpha_{i-1}$ are given and set
\begin{equation}
    X^{(i-1)} := \sum_{n=0}^{i-1} \hbar^n X_n, \qquad
    \alpha^{(i-1)} := \sum_{n=0}^{i-1} \hbar^n \alpha_n.
\label{X(i-1),a(i-1)}
\end{equation}
 \item (Step 1) Set
\begin{align}
    P^{(i-1)} 
    &:=
    \Ad \left(\exp \hbar^{-1} X^{(i-1)}\right) 
    \exp\left( \hbar^{-1} \alpha^{(i-1)} \log(\hbar\der)\right)
    f_t,
\label{def:Pi-1}
\\
    Q^{(i-1)}
    &:=
    \Ad \left(\exp \hbar^{-1} X^{(i-1)}\right) 
    \exp\left( \hbar^{-1} \alpha^{(i-1)} \log(\hbar\der)\right)
    g_t.
\label{def:Qi-1}
\end{align}
       Expand $P^{(i-1)}$ and $Q^{(i-1)}$ with respect to the
       $\hbar$-order as
\begin{align}
    P^{(i-1)} &=
    P^{(i-1)}_0 + \hbar P^{(i-1)}_1 + \cdots +
    \hbar^{k} P^{(i-1)}_k + \cdots
\label{Pi-1:h-expand}
\\
    Q^{(i-1)} &=
    Q^{(i-1)}_0 + \hbar Q^{(i-1)}_1 + \cdots +
    \hbar^{k} Q^{(i-1)}_k + \cdots.
\label{Qi-1:h-expand}
\end{align}
        ($\ordh P^{(i-1)}_k=\ordh Q^{(i-1)}_k=0$.)
 \item (Step 2) Put $\calP_0:=\symh(P^{(i-1)}_0)$,
       $\calQ_0:=\symh(Q^{(i-1)}_0)$,
       $\calP^{(i-1)}_i:=\symh(P^{(i-1)}_i)$,
       $\calQ^{(i-1)}_i:=\symh(Q^{(i-1)}_i)$ and define a constant
       $\alpha_i$ and a series $\tilde\calX_i(x,t,\xi)=\sum_{k=1}^\infty
       \tilde\chi_{i,k}(x,t)\xi^{-k}$ by
\begin{equation}
    \alpha_i \log\xi + \tilde\calX_i:=
    \int^\xi\left(
    \dfrac{\der \calQ_0}{\der \xi} \calP^{(i-1)}_i
    -
    \dfrac{\der \calP_0}{\der \xi} \calQ^{(i-1)}_i
    \right)_{\leq -1}
    d\xi.
\label{ai+tildeXi=int}
\end{equation}
       The integral constant of the indefinite integral is fixed so that
       the right hand side agrees with the left hand side.
 \item (Step 3) Define a series $\calX_i(x,t,\xi)=\sum_{k=1}^\infty
       \chi_{i,k}(x,t)\xi^{-k}$ by
\begin{equation}
 \begin{split}
    \calX_i
    &= \tilde \calX'_i
    - \frac{1}{2} \{\symh(X_0),\tilde\calX'_i\}
    + \sum_{p=1}^\infty K_{2p}
      (\ad_{\{,\}} (\symh(X_0)))^{2p} \tilde\calX'_i,
\\
    \tilde\calX'_i &:= 
    \alpha_i \log\xi + \tilde\calX_i(x,\xi)
    - \exp(\ad_{\{,\}} \symh(X_0)) (\alpha_i \log\xi).
 \end{split}
\label{tildeXi->Xi:symbol}
\end{equation}
       Here $K_{2p}$ is determined by the generating function
\begin{equation}
    \frac{z}{e^z-1}
    = 1 - \frac{z}{2} + \sum_{p=1}^\infty K_{2p} z^{2p},
\label{def:K2p}
\end{equation}
       or $K_{2p}= B_{2p}/(2p)!$, where $B_{2p}$'s are the Bernoulli
       numbers. 
 \item (Step 4) The operator $X_i(x,t,\hbar\der)$ is defined as the
       operator with the principal symbol $\calX_i$:
\begin{equation}
    X_i=\sum_{k=1}^\infty\chi_{i,k}(x,t)(\hbar\der)^{-k}.
\label{def:Xi}
\end{equation}
\end{itemize}

The main theorem is the following:
\begin{thm}
\label{thm:recursion:X}
 Assume that $X_0$ and $\alpha_0$ satisfy \eqref{sym(f,g)(L,M)-=0} and
 construct $X_i$'s and $\alpha_i$'s by the above procedure
 recursively. Then $X$ and $\alpha$ defined by \eqref{X:h-expansion}
 satisfy \eqref{f(M,L)=g(M,L)=diff}. Namely
 $W=\exp(X/\hbar)(\hbar\der)^{\alpha/\hbar}$ is a dressing operator of
 the $\hbar$-dependent KP hierarchy.
%
\end{thm}

\bigskip
The rest of this section is the proof of \thmref{thm:recursion:X} by
induction.

Let us denote the ``known'' part of $X$ and $\alpha$ by $X^{(i-1)}$ and
$\alpha^{(i-1)}$ as in \eqref{X(i-1),a(i-1)} and, as intermediate
objects, consider $P^{(i-1)}$ and $Q^{(i-1)}$ defined by
\eqref{def:Pi-1} and \eqref{def:Qi-1}, which are expanded as
\eqref{Pi-1:h-expand} and \eqref{Qi-1:h-expand}.

If $X$ and $\alpha$ are expanded as \eqref{X:h-expansion} and
\eqref{alpha:h-expansion}, the dressing operator
$W=\exp(X/\hbar)\exp(\alpha\log(\hbar\der)/\hbar)$ is factored as
follows by the Campbell-Hausdorff theorem:
\begin{multline}
    W=
    \exp\left(
     \hbar^{i-1} (\alpha_i \log(\hbar\der) + \tilde X_i) + 
     \hbar^i X_{>i}
    \right) \times
\\
    \times
    \exp\left( \hbar^{-1} X^{(i-1)} \right) 
    \exp\bigl( \hbar^{-1} \alpha^{(i-1)} \log(\hbar\der)\bigr),
\label{W=exp(Xi)exp(Xi-1)}
\end{multline}
where $\ordh(\alpha_i\log(\hbar\der)+\tilde X_i(x,\hbar\der))=0$,
$\ordh(X_{>i})\leqq 0$ and the principal symbol of
$\alpha_i\log(\hbar\der)+\tilde X_i(x,\hbar\der)$ is defined by
\begin{multline}
    \symh(\alpha_i\log(\hbar\der)+\tilde X_i){}(x,\xi)
\\
    =
    \sum_{n=1}^\infty \frac{(\ad_{\{,\}} \symh(X_0) )^{n-1}}{n!}
    \symh(X_i)
    +
    \exp\bigl(\ad_{\{,\}} \symh(X_0)\bigr) (\alpha_i \log\xi).
\label{def:tildeX}
\end{multline}
Note that the only log term in \eqref{def:tildeX} is $\alpha_i \log\xi$
and the rest is sum of negative powers of $\xi$. The principal symbol of
$X_i$ is recovered from $\tilde X_i$ by the formula
\begin{equation}
 \begin{split}
    \symh(X_i) 
    &= \symh(\tilde X'_i)
    - \frac{1}{2} \{\symh(X_0),\symh(\tilde X'_i)\}
    + \sum_{p=1}^\infty K_{2p}
      (\ad_{\{,\}} (\symh(X_0)))^{2p} \symh(\tilde X'_i), 
\\
    \symh(\tilde X'_i) &:= 
    \symh(\tilde X_i)(x,\xi)
    - \exp(\ad_{\{,\}} \symh(X_0)) (\alpha_i \log\xi)
\\
    &=
    \sum_{n=1}^\infty \frac{(\ad_{\{,\}} \symh(X_0) )^{n-1}}{n!}
    \symh(X_i).
 \end{split}
\label{tildeXi->Xi}
\end{equation}
Here coefficients $K_{2p}$ are defined by \eqref{def:K2p}. This
inversion relation is the origin of \eqref{tildeXi->Xi:symbol}. (Note
that the principal symbol determines the operator $X_i$, since it is a
homogeneous term in the expansion \eqref{X:h-expansion}.) We prove
formulae \eqref{W=exp(Xi)exp(Xi-1)} and \eqref{tildeXi->Xi} in
\appref{app:proof-campbell-hausdorff}.

The factorisation \eqref{W=exp(Xi)exp(Xi-1)} implies
\begin{equation*}
 \begin{split}
    P =& 
    \Ad\Bigl( \exp\bigl(
      \hbar^{i-1} (\alpha_i\log(\hbar\der)+\tilde X_i)
    + \hbar^i     X_{>i}
    \bigr) \Bigr) P^{(i-1)}
\\
    =&
    P^{(i-1)} 
    + \hbar^{i-1}
    [(\alpha_i\log(\hbar\der)+\tilde X_i)+\hbar X_{>i}, P^{(i-1)}]
    + (\text{terms of $\hbar$-order $<-i$}).
 \end{split}
\end{equation*}
Thus, substituting the expansion \eqref{Pi-1:h-expand} in the step 1, we
have
\begin{equation}
 \begin{split}
    P =& 
    P^{(i-1)}_0 + \hbar P^{(i-1)}_1 + \cdots + \hbar^i P^{(i-1)}_i +
    \cdots
\\
    &+ \hbar^{i-1} [\alpha_i\log(\hbar\der)+\tilde X_i,
         P^{(i-1)}_0]
\\
    &+ (\text{terms of $\hbar$-order $<-i$}).
 \end{split}
\label{P=Ad()Pi-1}
\end{equation}
Comparing this with the $\hbar$-expansion of $P$ \eqref{P}, we can
express $P_i$'s in terms of $P^{(i-1)}_j$, $\tilde X_i$ and $\alpha_i$
as follows:
\begin{align}
    P_j &= P^{(i-1)}_j \qquad  (j=0,\dots,i-1),
\label{P0-Pi-1}
\\
    \sigma_0 (P_i) &= \sigma_0 (P^{(i-1)}_i 
    + h^{-1}[\alpha_i\log(\hbar\der)+\tilde X_i, P^{(i-1)}_0]).
\label{Pi<-Pi-1i}
\end{align}
Similar equations for $Q$ are obtained in the same way. The first
equations \eqref{P0-Pi-1} show that the terms of $\hbar$-order greater
than $-i$ in \eqref{P} are already fixed by $X_0,\dots,X_{i-1}$ and
$\alpha_0,\dots,\alpha_{i-1}$, which justifies the inductive
procedure. That is to say, we are assuming that $X_0,\dots,X_{i-1}$ and
$\alpha_0,\dots,\alpha_{i-1}$ have been already determined so that
$P_j=P^{(i-1)}_j$ and $Q_j=Q^{(i-1)}_j$ for $j=0,\dots,i-1$ are
differential operators.

The operator $X_i$ and constant $\alpha_i$ should be chosen so that the
right hand side of \eqref{Pi<-Pi-1i} and the corresponding expression
for $Q$ are differential operators. Taking equations $P^{(i-1)}_0=P_0$
and $Q^{(i-1)}_0=Q_0$ into account, we define
\begin{equation}
 \begin{split}
    \tilde P^{(i)}_i &:= P^{(i-1)}_i 
    + \hbar^{-1} [\alpha_i\log(\hbar\der)+\tilde X_i, P_0],
\\
    \tilde Q^{(i)}_i &:= Q^{(i-1)}_i 
    + \hbar^{-1} [ \alpha_i\log(\hbar\der)+\tilde X_i, Q_0].
 \end{split}
\label{Pii,Qii}
\end{equation}
Then the condition for $X_i$ and $\alpha_i$ is written in the following
form of equations for symbols:
\begin{equation}
    (\symh_0(\tilde P^{(i)}_i))_{\leq -1} = 0, \qquad
    (\symh_0(\tilde Q^{(i)}_i))_{\leq -1} = 0.
\label{Pii,Qii:diff-op}
\end{equation}
(The parts of $\hbar$-order less than $-1$ should be determined in the
next step of the induction.) To simplify notations, we denote the
symbols $\symh_0(\tilde P^{(i)}_i)$, $\symh_0(P^{(i-1)}_i)$ and so on by
the corresponding calligraphic letters as $\calP^{(i)}_i$,
$\calP^{(i-1)}_i$ etc. By this notation we can rewrite the equations
\eqref{Pii,Qii:diff-op} in the following form:
\begin{equation}
 \begin{aligned}
    (\tilde\calP^{(i)}_i)_{\leq -1} &= 0, \qquad&
    \tilde\calP^{(i)}_i &:= \calP^{(i-1)}_i +
    \{ \alpha_i \log\xi + \tilde\calX_i, \calP_0\},
\\
    (\tilde\calQ^{(i)}_i)_{\leq -1} &= 0, \qquad&
    \tilde\calQ^{(i)}_i &:= \calQ^{(i-1)}_i +
    \{ \alpha_i \log\xi + \tilde\calX_i, \calQ_0\}.
  \end{aligned}
\label{Pii,Qii:symbol}
\end{equation}
The above definitions of $\tilde\calP^{(i)}_i$ and $\tilde\calQ^{(i)}_i$
are written in the matrix form:
\begin{equation}
    \begin{pmatrix}
    \dfrac{\der \calP_0}{\der x} & - \dfrac{\der \calP_0}{\der \xi} \\ \\
    \dfrac{\der \calQ_0}{\der x} & - \dfrac{\der \calQ_0}{\der \xi} 
    \end{pmatrix}
    \begin{pmatrix}
    \dfrac{\der}{\der \xi} (\alpha_i \log\xi + \tilde\calX_i) \\ \\
    \dfrac{\der}{\der x}   (\alpha_i \log\xi + \tilde\calX_i)
    \end{pmatrix}
    =
    \begin{pmatrix}
    \tilde\calP^{(i)}_i - \calP^{(i-1)}_i \\ \\
    \tilde\calQ^{(i)}_i - \calQ^{(i-1)}_i
    \end{pmatrix}.
\label{Pii,Qii:mat}
\end{equation}
Recall that operators $P^{(i-1)}$ and $Q^{(i-1)}$ are defined by acting
adjoint operation to the canonically commuting pair $(f,g)$ in
\eqref{def:Pi-1}, \eqref{def:Qi-1} and \eqref{ft,gt}. Hence they also
satisfy the canonical commutation relation:
$[P^{(i-1)},Q^{(i-1)}]=\hbar$. The principal symbol of this relation
gives
\[
   \{\calP^{(i-1)}_0,\calQ^{(i-1)}_0\}=\{\calP_0,\calQ_0\}=1,
\]
which means that the determinant of the matrix in the left hand side of
\eqref{Pii,Qii:mat} is equal to $1$. Hence its inverse matrix is easily
computed and we have
\begin{equation}
    \begin{pmatrix}
    \dfrac{\der}{\der \xi} (\alpha_i \log\xi + \tilde\calX_i) \\ \\
    \dfrac{\der}{\der x}   (\alpha_i \log\xi + \tilde\calX_i)
    \end{pmatrix}
    =
    \begin{pmatrix}
    - \dfrac{\der \calQ_0}{\der \xi}  & \dfrac{\der \calP_0}{\der \xi}
    \\ \\
    - \dfrac{\der \calQ_0}{\der x}    & \dfrac{\der \calP_0}{\der x}
    \end{pmatrix}
    \begin{pmatrix}
    \tilde\calP^{(i)}_i - \calP^{(i-1)}_i \\ \\ 
    \tilde\calQ^{(i)}_i - \calQ^{(i-1)}_i
    \end{pmatrix}.
\label{alpha,X:mat}
\end{equation}
We are assuming that $\calP_0$ and $\calQ_0$ do not contain negative
powers of $\xi$ and we are searching for $\alpha_i\log\xi+\tilde\calX_i$
such that $\tilde\calP^{(i)}_i$ and $\tilde\calQ^{(i)}_i$ are series of
$\xi$ 
without negative powers. Since $\alpha_i$ is constant with respect to
$x$, the left hand side of \eqref{alpha,X:mat} contain only negative
powers of $\xi$. Thus taking the negative power parts of the both hand
sides in \eqref{alpha,X:mat}, we have
\begin{equation}
    \begin{pmatrix}
    \dfrac{\der}{\der \xi} (\alpha_i \log\xi + \tilde\calX_i) \\ \\
    \dfrac{\der}{\der x}   (\alpha_i \log\xi + \tilde\calX_i)
    \end{pmatrix}
    =
    \begin{pmatrix}
    \left(
    \dfrac{\der \calQ_0}{\der \xi} \calP^{(i-1)}_i
    -
    \dfrac{\der \calP_0}{\der \xi} \calQ^{(i-1)}_i
    \right)_{\leq -1}
    \\ \\
    \left(
    \dfrac{\der \calQ_0}{\der x}   \calP^{(i-1)}_i
    -
    \dfrac{\der \calP_0}{\der x}   \calQ^{(i-1)}_i
    \right)_{\leq -1}
    \end{pmatrix}.
\label{dXi/dxi,dXi/dx}
\end{equation}
This is the equation which determines $\alpha_i$ and $X_i$. 

The system \eqref{dXi/dxi,dXi/dx} is solvable thanks to
\lemref{lem:compatibility} below. Hence, integrating the first element
of the right hand side with respect to $\xi$, we obtain $\alpha_i
\log\xi + \tilde\calX_i$. This is Step 2, \eqref{ai+tildeXi=int}. In the
end, the principal symbol of $X_i$ is determined by \eqref{tildeXi->Xi}
or \eqref{tildeXi->Xi:symbol} in Step 3 and thus $X_i$ is defined as in
Step 4. This completes the construction of $X_i$ and $\alpha_i$ and the
proof of the theorem.
%

\begin{lem}
\label{lem:compatibility}
 The system \eqref{dXi/dxi,dXi/dx} is compatible.
\end{lem}

\begin{proof}
 We check:
\begin{equation}
 \begin{split}
    &\frac{\der}{\der x}
    \left(
    \dfrac{\der \calQ_0}{\der \xi} \calP^{(i-1)}_i
    -
    \dfrac{\der \calP_0}{\der \xi} \calQ^{(i-1)}_i
    \right)_{\leq -1}
\\
    =&
    \frac{\der}{\der \xi}
    \left(
    \dfrac{\der \calQ_0}{\der x}   \calP^{(i-1)}_i
    -
    \dfrac{\der \calP_0}{\der x}   \calQ^{(i-1)}_i
    \right)_{\leq -1}.
 \end{split}
\label{compatibility}
\end{equation}
Since differentiation commutes with truncation of power series, condition
\eqref{compatibility} is equivalent to saying that the negative power
part of the following is zero:
\begin{equation}
 \begin{split}
    &\frac{\der}{\der x}
    \left(
    \dfrac{\der \calQ_0}{\der \xi} \calP^{(i-1)}_i
    -
    \dfrac{\der \calP_0}{\der \xi} \calQ^{(i-1)}_i
    \right)
    -
    \frac{\der}{\der \xi}
    \left(
    \dfrac{\der \calQ_0}{\der x}   \calP^{(i-1)}_i
    -
    \dfrac{\der \calP_0}{\der x}   \calQ^{(i-1)}_i
    \right)
\\
    =&
    \dfrac{\der \calQ_0}{\der \xi} \frac{\der \calP^{(i-1)}_i}{\der x}
    -
    \dfrac{\der \calP_0}{\der \xi} \frac{\der \calQ^{(i-1)}_i}{\der x}
    -
    \dfrac{\der \calQ_0}{\der x}   \frac{\der \calP^{(i-1)}_i}{\der \xi}
    +
    \dfrac{\der \calP_0}{\der x}   \frac{\der \calQ^{(i-1)}_i}{\der \xi}
\\  
    =& - \{ \calP^{(i-1)}_i, \calQ_0 \}
       - \{ \calP_0, \calQ^{(i-1)}_i \}.
 \end{split}
\label{compatibility2}
\end{equation}
Defined from canonically commuting pair $(f,g)$ by adjoint action
\eqref{def:Pi-1} and \eqref{def:Qi-1}, the pair of operators $(P^{(i-1)},
Q^{(i-1)})$ is canonically commuting: $[P^{(i-1)},Q^{(i-1)}]=\hbar$. The
negative order part of this relation is zero. On the other hand,
substituting the expansions $P^{(i-1)}=\sum_{n=0}^\infty \hbar^n
P^{(i-1)}_n$ and $Q^{(i-1)}=\sum_{n=0}^\infty \hbar^n Q^{(i-1)}_n$ in
this canonical commutation relation and noting that $P^{(i-1)}_j$ and
$Q^{(i-1)}_j$ ($j=0,\dots,i-1$) do not contain negative order part by
the induction hypothesis, we have
\begin{equation*}
 \begin{split}
    0 &= ([P^{(i-1)}, Q^{(i-1)}])_{\leq -1}
\\
    &=
    [\hbar^i P^{(i-1)}_i, Q^{(i-1)}_0] +
    [P^{(i-1)}_0, \hbar^i Q^{(i-1)}_i] +
    (\text{terms of $\hbar$-order $<-i-1$}).
 \end{split}
\end{equation*}
Taking the symbol of $\hbar$-order $-i-1$ of this equation, we have
\begin{equation*}
    0 = \{ \calP^{(i-1)}_i, \calQ_0 \}
      + \{ \calP_0, \calQ^{(i-1)}_i\},
\end{equation*}
which proves that \eqref{compatibility2} vanishes.
\end{proof}

\section{Asymptotics of the wave function}
\label{sec:wave-function}

In this section we prove that the dressing operator of the form
\begin{equation}
    W(\hbar,x,t, \hbar\der) = \exp(X(\hbar,x,\hbar\der)/\hbar), \qquad
    \ordh X \leqq 0, \qquad \ord X \leqq -1,
\label{dressing}
\end{equation}
gives a wave function of the form
\begin{gather}
    \Psi(\hbar,x,t;z) 
    = W e^{(xz + \zeta(t,z))/\hbar} = \exp(S(\hbar,x,t,z)/\hbar),
    \qquad \ordh S \leqq 0,
\label{wave-func}
\\
    S(\hbar,x,t;z)
    = \sum_{n=0}^\infty \hbar^n S_n(x,t;z) + \zeta(t,z), \qquad
    \zeta(t,z) := \sum_{n=1}^\infty t_n z^n,
\label{S}
\end{gather}
and vice versa. 

Since the time variables $t_n$ do not play any role in this section,
we set them to zero. As the factor $(\hbar\der)^{\alpha/\hbar}$ in
\eqref{W=exp(X)} becomes a constant factor $z^{\alpha/\hbar}$ when it is
applied to $e^{xz/\hbar}$, we also omit it here.

Let $A(\hbar,x,\hbar\der)=\sum_n a_n(\hbar,x) (\hbar\der)^n$ be a
microdifferential operator. The {\em total symbol} of $A$ is a power
series of $\xi$ defined by
\begin{equation}
    \totsym(A)(\hbar,x,\xi):= \sum_n a_n(\hbar,x) \xi^n.
\label{tot-symbol}
\end{equation}
Actually, this is the factor which appears when the operator $A$ is
applied to $e^{xz/\hbar}$:
\begin{equation}
    A e^{xz/\hbar} = \totsym(A)(\hbar,x,z) e^{xz/\hbar}.
\label{tot-symbol:Laplace}
\end{equation}
Using this terminology, what we show in this section is that a operator
of the form $e^{X/\hbar}$ has a total symbol of the form $e^{S/\hbar}$
and that an operator with total symbol $e^{S/\hbar}$ has a form
$e^{X/\hbar}$. Exactly speaking, the main results in this section are
the following two propositions.
\begin{prop}
\label{prop:exp(:X:)=:exp(S):}
 Let $X=X(\hbar,x,\hbar\der)$ be a microdifferential operator such that
 $\ord X=-1$ and $\ordh X = 0$. Then the total symbol of $e^{X/\hbar}$
 has such a form as
\begin{equation}
    \totsym(\exp(\hbar^{-1} X(\hbar,x,\hbar\der)))
    = e^{S(\hbar,x,\xi)/\hbar},
\label{exp(:X:)=:exp(S):}
\end{equation}
 where $S(\hbar,x,\xi)$ is a power series of $\xi^{-1}$ without
 non-negative powers of $\xi$ and has an $\hbar$-expansion
\begin{equation}
    S(\hbar,x,\xi)=\sum_{n=0}^\infty \hbar^n S_n(x,\xi).
\label{S:h-expansion}
\end{equation}

 Moreover, the coefficient $S_n$ is determined by $X_0,\dots,X_n$ in the
 $\hbar$-expansion \eqref{X:h-expansion} of $X=\sum_{n=0}^\infty \hbar^n
 X_n$.
\end{prop}

Explicitly, $S_n$ is determined as follows:
\begin{itemize}
 \item (Step 0) Assume that $X_0,\dots,X_n$ are given. Let $X_i(x,\xi)$
       be the total symbol $\totsym(X_i(x,\hbar\der))$.

 \item (Step 1) Define $Y^{(l)}_{k,m}(x,y,\xi,\eta)$ and
       $S^{(l)}(x,\xi)$ by the following recursion relations:
\begin{align}
    &Y^{(l)}_{k,-1}=0
\\
    &S^{(0)}_m=0,
\\
    &Y^{(l)}_{0,m}(x,y,\xi,\eta) = \delta_{l,0} X_m(x,\xi)
\label{Yl0m}
\end{align}
for $l\geqq 0$, $m=0,\dots,n$,
\begin{multline}
    Y^{(l)}_{k+1,m}(x,y,\xi,\eta)
\\    =
    \frac{1}{k+1}
    \left(
    \der_\xi \der_y Y^{(l)}_{k,m-1}(x,y,\xi,\eta)
    +
    \sum_{\substack{0\leq l' \leq l-1 \\ 0\leq m' \leq m}}
    \der_\xi Y^{(l')}_{k,m'}(x,y,\xi,\eta)
    \der_y S^{(l-l')}_{m-m'}(y,\eta)
    \right)
\label{recursion:Y(l,k,m)}
\end{multline}
       for $k\geqq 0$, and
\begin{equation}
    S^{(l+1)}_m(x,\xi)
    = \frac{1}{l+1} \sum_{k=0}^{l+m} Y^{(l)}_{k,m}(x,x,\xi,\xi).
\label{recursion:S(l+1,m)}
\end{equation}
       (We shall prove that $Y^{(l)}_{k,m}=0$ if $k>l+m$.) Schematically
       this procedure goes as follows:
{\small
\begin{equation*}
 \begin{matrix}
                     &         & Y^{(l)}_{0,0}=\delta_{l,0}X_0
                     &         & Y^{(l)}_{0,1}=\delta_{l,0}X_1
                     &         & Y^{(l)}_{0,2}=\delta_{l,0}X_2
\\
                     &         & +
                     &\searrow & +
                     &\searrow & +
\\
    Y^{(l)}_{k,-1}=0 &\to      & Y^{(l)}_{k,0}
                     &\to      & Y^{(l)}_{k,1}
                     &\to      & Y^{(l)}_{k,2} & \cdots
\\
                     &         & \downarrow
                     &\nearrow & \downarrow
                     &\nearrow & \downarrow
\\
                     &         & S^{(l+1)}_0
                     &         & S^{(l+1)}_1
                     &         & S^{(l+1)}_2
 \end{matrix}
\end{equation*}
}

 \item (Step 2) $S_n(x,\xi)= \sum_{l=1}^\infty S^{(l)}_n(x,\xi)$. (The
       sum makes sense as a power series of $\xi$.)
\end{itemize}

\begin{prop}
\label{prop::exp(S):=exp(:X:)}
 Let $S=\sum_{n=0}^\infty \hbar^n S_n$ be a power series of $\xi^{-1}$
 without non-negative powers of $\xi$. Then there exists a
 microdifferential operator $X(\hbar,x,\hbar\der)$ such that $\ord X
 \leqq -1$, $\ordh X \leqq 0$ and
\begin{equation}
    \totsym(\exp(\hbar^{-1} X(\hbar,x,\hbar\der)))
    = e^{S(\hbar,x,\xi)/\hbar}.
\label{:exp(S):=exp(:X:)}
\end{equation}
 Moreover, the coefficient $X_n(x,\xi)$ in the $\hbar$-expansion
 $X=\sum_{n=0}^\infty \hbar^n X_n$ of the total symbol
 $X=X(\hbar,x,\xi)$ is determined by $S_0,\dots,S_n$ in the
 $\hbar$-expansion of $S$.
 
\end{prop}

Explicit procedure is as follows:
\begin{itemize}
 \item (Step 0) Assume that $S_0,\dots,S_n$ are given. Expand them into
       homogeneous terms with respect to powers of $\xi$:
       $S_n(x,\xi)=\sum_{j=1}^\infty S_{n,j}(x,\xi)$, where $S_{n,j}$ is
       a term of degree $-j$.
 \item (Step 1) Define $Y^{(l)}_{k,n,j}(x,y,\xi,\eta)$ as follows:
\begin{align}
    &Y^{(l)}_{k,-1,j}(x,y,\xi,\eta)=0,
\\
    &Y^{(l)}_{k,m,1}(x,y,\xi,\eta)
    =\delta_{l,0}\delta_{k,0} S_{m,1}(x,\xi)
\label{Ylkm1}
\end{align}
       for $m=0,\dots,n$, $k\geqq 0$, $l\geqq 0$ and
\begin{equation}
    Y^{(l)}_{0,m,j}=0
\end{equation}
       for $m=0,\dots,n$, $l>0$, $j\geqq 1$. For other $(l,k,m,j)$,
       $(l,k)\neq(0,0)$, $Y^{(l)}_{k,m,j}$ are determined by the
       recursion relation:
\begin{multline}
    Y^{(l)}_{k+1,m,j}(x,y,\xi,\eta)
    =
    \frac{1}{k+1}
    \Biggl(
    \der_\xi \der_y Y^{(l)}_{k,m-1,j-1}(x,y,\xi,\eta)+
\\
    +
    \sum_{\substack{0\leq l' \leq l-1\\
                    0\leq j' \leq j-1, 
                    0\leq m' \leq m \\
                    0\leq k'' \leq j-j'-l+l'}}
    \frac{1}{l-l'}
    \der_\xi Y^{(l')}_{k,m',j'}(x,y,\xi,\eta)
    \der_y Y^{(l-l'-1)}_{k'',m-m',j-j'-1}(x,x,\xi,\xi)
    \Biggr).
\label{recursion:Y(l,k,m)<-Y}
\end{multline}
       The remaining $Y^{(0)}_{0,m,j}$ is determined by:
\begin{equation}
    Y^{(0)}_{0,m,j}(x,y,\xi,\eta)
    =
    S_{m,j}(x,\xi)
    - \sum_{\substack{(l,k)\neq(0,0)\\l,k\geq 0, l+k\leq j}}
      \frac{1}{l+1} Y^{(l)}_{k,m,j}(x,x,\xi,\xi).
\label{Y00mj=Smj-Ylkmj}
\end{equation}
       (We shall show that $Y^{(l)}_{k,m,j}=0$ for $l+k>j$.)
       Schematically this procedure goes as follows:
\begin{equation*}
 \begin{matrix}
    &  & 
    Y^{(l)}_{k,m,1}=\delta_{l,0}\delta_{k,0}S_{m,1}
\\
    &  & 
    \downarrow\hskip 1.5cm 
\\
    Y^{(l')}_{k',m',1}(m'<m)      & \to  & 
    Y^{(l)}_{k,m,2}\ (k,l\neq 0)  & \to  &
    Y^{(0)}_{0,m,2}               & \leftarrow &
    S_{m,2}
\\
    &  & 
    \downarrow\hskip 1.5cm & \swarrow
\\
    Y^{(l')}_{k',m',1}, Y^{(l')}_{k',m',2}(m'<m)      & \to  & 
    Y^{(l)}_{k,m,3}\ (k,l\neq 0)  & \to  &
    Y^{(0)}_{0,m,3}               & \leftarrow &
    S_{m,3}
\\
    &  & 
    \vdots\hskip 1.5cm
 \end{matrix}
\end{equation*}

\item (Step 2) $X_n(x,\xi)=\sum_{j=1}^\infty
       Y^{(0)}_{0,n,j}(x,x,\xi,\xi)$.  (The infinite sum is the
       homogeneous expansion in terms of powers of $\xi$.)
\end{itemize}

Combining these propositions with the results in \secref{sec:recursion},
we can, in principle, make a recursion formula for $S_n$
($n=0,1,2,\dots$) of the wave function of the solution of the KP
hierarchy corresponding to the quantised canonical transformation
$(f,g)$ as follows: let $S_0,\dots,S_{i-1}$ be given.
\begin{enumerate}
 \item By \propref{prop::exp(S):=exp(:X:)} we have $X_0,\dots,X_{i-1}$.
 \item We have a recursion formula for $X_i$ by
       \thmref{thm:recursion:X}. 
 \item \propref{prop:exp(:X:)=:exp(S):} gives a formula for $S_i$.
\end{enumerate}
If we take the factor $(\hbar\der)^{\alpha/\hbar}$ into account, this
process becomes a little bit complicated, but essentially the same.

\bigskip
The rest of this section is devoted to the proof of
\propref{prop:exp(:X:)=:exp(S):} and \propref{prop::exp(S):=exp(:X:)}. 

If one would consider that ``the principal symbol is obtained just by
replacing $\hbar\der$ with $\xi$'', the statements of the above
propositions might seem obvious, but since the multiplication in the
definition of
\begin{equation}
    e^{X/\hbar} 
    = \sum_{n=0}^\infty \frac{X(\hbar,x,\hbar\der)^n}{\hbar^n n!}
\label{exp(X)}
\end{equation}
is non-commutative, while the multiplication of total symbols in the
series 
\begin{equation}
    e^{S/\hbar} 
    = \sum_{n=0}^\infty \frac{S(\hbar,x,\xi)^n}{\hbar^n n!}
\label{exp(S)}
\end{equation}
is commutative, it needs to be proved. In fact, computation of the
simplest example of $X=x (\hbar\der)^{-1}$ would show how complicated
the formula can be:
\begin{equation*}
 \begin{split}
    & \totsym(e^{x(\hbar\der)^{-1}/\hbar}) =
    \sum_{n=0}^\infty \frac{1}{n!\hbar^n} \totsym(X^n)
\\
    =&
    1 + \frac{1}{1! \hbar} x\xi^{-1}
    + \frac{1}{2! \hbar^2} (x^2 \xi^{-2} - \hbar x \xi^{-3})
\\
    &+ \frac{1}{3! \hbar^3} 
      (x^3 \xi^{-3} - 3 \hbar x^2 \xi^{-4} + 3 \hbar^2 x \xi^{-5})
\\
    &+ \frac{1}{4! \hbar^4}
      (x^4 \xi^{-4} - 6 \hbar x^3 \xi^{-5} + 15 \hbar^2 x^2 \xi^{-6}
       - 15 \hbar^3 x \xi^{-7})
    + \cdots
\\
    =&
    \exp 
    \frac{1}{\hbar}\left(
     x\xi^{-1} - \frac{x\xi^{-3}}{2} + \frac{x\xi^{-5}}{2}
     - \frac{5x\xi^{-7}}{8} + \cdots
    \right).
 \end{split}
\end{equation*}
It is not obvious, at least to authors, why there is no more negative
powers of $\hbar$ in the last expression, which can be obtained by
calculating the logarithm of the previous expression.

To avoid confusion, the commutative multiplication of total symbols
$a(\hbar,x,\xi)$ and $b(\hbar,x,\xi)$ as power series is denoted by
$a(\hbar,x,\xi)\, b(\hbar,x,\xi)$ and the non-commutative multiplication
corresponding to the operator product is denoted by $a(\hbar,x,\xi)
\circ b(\hbar,x,\xi)$. Recall that the latter multiplication is
expressed (or defined) as follows:
\begin{equation}
 \begin{split}
    a(\hbar,x,\xi) \circ b(\hbar,x,\xi)
    &=
    e^{\hbar\der_\xi \der_y} 
    a(\hbar,x,\xi) b(\hbar,y,\eta)|_{y=x,\eta=\xi}
\\
    &=
    \sum_{n=0}^\infty\frac{\hbar^n}{n!}
    \der_\xi^n a(\hbar,x,\xi) \der_y^n b(\hbar,y,\eta)|_{y=x,\eta=\xi}.
\label{def:symbol-prod}
 \end{split}
\end{equation}
(See, for example, \cite{sch:85}, \cite{aok:86} or \cite{kas-rou:08}.)
The order of an operator corresponding to symbol $a(\hbar,x,\xi)$ is
denoted by $\ordx a(\hbar,x,\xi)$, which is the same as the order of
$a(\hbar,x,\xi)$ as a power series of $\xi$. The $\hbar$-order is the
same as that of the operator: $\ordh x = \ordh \xi = 0$, $\ordh \hbar =
-1$.

The main idea of proof of propositions is due to Aoki \cite{aok:86},
where exponential calculus of pseudodifferential operators is
considered. He considered analytic symbols of exponential type, while
our symbols are formal ones. Therefore we have only to confirm that
those symbols make sense as formal series.

First, we prove the following lemma.
\begin{lem}
\label{lem:prod-exp-symbol}
 Let $a(\hbar,x,\xi)$, $b(\hbar,x,\xi)$, $p(\hbar,x,\xi)$ and
 $q(\hbar,x,\xi)$ be symbols such that
 $\ordx a(\hbar,x,\xi) = M$, $\ordh a(\hbar,x,\xi) = 0$, $\ordx
 b(\hbar,x,\xi) = N$, $\ordh b(\hbar,x,\xi) = 0$, $\ordx p(\hbar,x,\xi)
 = \ordx q(\hbar,x,\xi) = 0$, $\ordh p(\hbar,x,\xi) = \ordh
 q(\hbar,x,\xi) = 0$.

 Then there exist symbols $c(\hbar,x,\xi)$ ($\ordx c(\hbar,x,\xi) =
 N+M$, $\ordh c(\hbar,x,\xi) = 0$) and $r(\hbar,x,\xi)$ ($\ordx
 r(\hbar,x,\xi) = 0$, $\ordh r(\hbar,x,\xi) = 0$) such that
\begin{equation}
    \bigl(a(\hbar,x,\xi) e^{p(\hbar,x,\xi)/\hbar}\bigr)
    \circ
    \bigl(b(\hbar,x,\xi) e^{q(\hbar,x,\xi)/\hbar}\bigr)
    =
    c(\hbar,x,\xi) e^{r(\hbar,x,\xi)/\hbar}.
\label{prod-exp-symbol}
\end{equation}
\end{lem}

In the proof of \propref{prop:exp(:X:)=:exp(S):} and
\propref{prop::exp(S):=exp(:X:)}, we use the construction of $c$ and $r$
in the proof of \lemref{lem:prod-exp-symbol}.

\begin{proof}
 Following \cite{aok:86}, we introduce a parameter $t$ and consider
\begin{equation}
    \pi(t) = \pi(t;\hbar,x,y,\xi,\eta):=
    e^{\hbar t \der_\xi \der_y}
    a(\hbar,x,\xi) b(\hbar,y,\eta)
    e^{\bigl(p(\hbar,x,\xi) + q(\hbar,y,\eta)\bigr)/\hbar}.
\label{def:pi(t)}
\end{equation}
 If we set $t=1$, $y=x$ and $\eta=\xi$, this reduces to the operator
 product of \eqref{def:symbol-prod}. The series $\pi(t)$ satisfies a
 differential equation with respect to $t$:
\begin{equation}
    \der_t \pi = \hbar \der_\xi\der_y \pi, \qquad
    \pi(0) = a(\hbar,x,\xi) b(\hbar,y,\eta) 
    e^{\bigl(p(\hbar,x,\xi) + q(\hbar,y,\eta)\bigr)/\hbar}.
\label{diff-eq:pi}
\end{equation}
 We construct its solution in the following form:
\begin{equation}
 \begin{aligned}
    \pi(t) &= \psi(t) e^{w(t)/\hbar},
\\
    \psi(t) &= \psi(t;\hbar,x,y,\xi,\eta)
             = \sum_{n=0}^\infty \psi_n t^n,
\\
    w(t) &= w(t;\hbar,x,y,\xi,\eta) = \sum_{k=0}^\infty w_k t^k.
 \end{aligned}
\label{pi=psi*ew}
\end{equation}
 Later we set $t=1$ and prove that $\psi(1)$ and $w(1)$ is meaningful as
 a formal power series of $\xi$ and $\eta$. The differential equation
 \eqref{diff-eq:pi} is rewritten as
\begin{equation}
 \begin{split}
    &\frac{\der\psi}{\der t} + \hbar^{-1} \psi \frac{\der w}{\der t}
\\
    =&
    \hbar \der_\xi \der_y \psi 
    + \der_\xi \psi \der_y w + \der_y \psi \der_\xi w
    +
    \psi 
    \left( \der_\xi \der_y w + \hbar^{-1} \der_\xi w \der_y w\right).
 \end{split}
\label{diff-eq:psi,w}
\end{equation}
 Hence it is sufficient to construct
 $\psi(t)=\psi(t;\hbar,x,y,\xi,\eta)$ and $w(t)=w(t;\hbar,x,y,\xi,\eta)$
 which satisfy 
\begin{align}
    \frac{\der w}{\der t} &= \hbar \der_\xi \der_y w
    + \der_\xi w \der_y w,
\label{diff-eq:w}
\\
    \frac{\der \psi}{\der t}
    &= \hbar \der_\xi \der_y \psi 
     + \der_\xi \psi \der_y w + \der_y \psi \der_\xi w.
\label{diff-eq:psi}
\end{align}
 (This is a {\em sufficient} condition but not a necessary condition for
 $\pi = \psi e^w$ to be a solution of \eqref{diff-eq:pi}. The solution
 of \eqref{diff-eq:pi} is unique, but $\psi$ and $w$ satisfying
 \eqref{diff-eq:psi,w} are not unique at all.)

 To begin with, we solve \eqref{diff-eq:w} and determine
 $w(t)$. Expanding it as $w(t) = \sum_{k=0}^\infty w_k t^k$, we have
 a recursion relation and the initial condition
\begin{equation}
 \begin{split}
    w_{k+1} &= \frac{1}{k+1} \left(
    \hbar \der_\xi \der_y w_k +
    \sum_{\nu=0}^k \der_\xi w_\nu \der_y w_{k-\nu}
    \right),
\\
    w_0 &= p(x,\xi) + q(y,\eta),
 \end{split}
\label{recursion:w}
\end{equation}
 which determine $w_k=w_k(\hbar,x,y,\xi,\eta)$ inductively. Note that
 $\ordx w_0\leqq 0$ and $\der_\xi$ lowers the order by one, which
 implies
\begin{equation}
    \ordx w_k \leqq -k.
\label{ordx(wk)}
\end{equation}
 (Here $\ordx$ denotes the order with respect to $\xi$ and $\eta$.) This
 shows that $w(1) = \sum_{k=0}^\infty w_k$ makes sense as a formal
 series of $\xi$ and $\eta$. Moreover it is obvious that $w_k$ and
 $w(1)$ are formally regular with respect to $\hbar$.

 As a next step, we expand $\psi(t)$ as $\psi(t) = \sum_{k=0}^\infty
 \psi_k t^k$ and rewrite \eqref{diff-eq:psi} into a recursion relation
 and the initial condition:
\begin{equation}
 \begin{split}
    \psi_{k+1} &= \frac{1}{k+1}
    \left(
     \hbar \der_\xi \der_y \psi_k +
     \sum_{\nu=0}^k (\der_\xi \psi_\nu \der_y w_{k-\nu}
                   + \der_y \psi_\nu \der_\xi w_{k-\nu})
    \right),
\\
    \psi_0 &= a(x,\xi) b(y,\eta).
 \end{split}
\label{recursion:psi}
\end{equation}
 In this case we have estimate of the order of the terms:
\begin{equation}
     \ordx \psi_k \leqq N+M-k,
\label{ordx:psik}
\end{equation}
 which shows that the inifinite sum $\psi(1) = \sum_{k=0}^\infty \psi_k$
 makes sense. The regularity of $\psi_k$ and $\psi(1)$ is also obvious.

 Thus, we have constructed $\pi(t) = \pi(t;\hbar,x,y,\xi,\eta) =
 \psi(t;\hbar,x,y,\xi,\eta) e^{w(t;\hbar,x,y,\xi,\eta)}$, which is
 meaningful also at $t=1$. Hence the product $a(\hbar,x,\xi) \circ
 b(\hbar,x,\xi) = \pi(1; \hbar,x,x,\xi,\xi)$ is expressed in the form
 $c(\hbar,x,\xi) e^{r(\hbar,x,\xi)/\hbar}$, where
 $c(\hbar,x,\xi)=\psi(1;\hbar,x,x,\xi,\xi)$, $r(\hbar,x,\xi)=
 w(1;\hbar,x,x,\xi,\xi)/\hbar$.
\end{proof}

\begin{proof}[Proof of \propref{prop:exp(:X:)=:exp(S):}]
 We make use of differential equations satisfied by the operator
\begin{equation}
    E(s)=E(s;\hbar,x,\hbar\der) :=
    \exp\left(\frac{s}{\hbar}X(\hbar,x,\hbar\der))\right),
\label{def:E(s)}
\end{equation}
 depending on a parameter $s$. The total symbol of $E(s)$ is defined as 
\begin{equation}
    E(s;\hbar,x,\xi) 
    = \sum_{k=0}^\infty \frac{s^k}{\hbar^k k!} X^{(k)}(\hbar,x,\xi),
    \qquad
    X^{(0)}=1, \qquad X^{(k+1)} = X \circ X^{(k)}.
\label{E:symbol}
\end{equation}
 Taking the logarithm (as a function, not as an operator) of this, we
 can define $S(s)=S(s;\hbar,x,\xi)$ by
\begin{equation}
    E(s;\hbar,x,\xi) =
    e^{\hbar^{-1} S(s;\hbar,x,\xi)}.
\label{E=eS}
\end{equation}
 What we are to prove is that $S(s)$, constructed as a series, makes
 sense at $s=1$ and formally regular with respect to $\hbar$.

 Differentiating \eqref{E=eS}, we have
\begin{equation}
    X(\hbar,x,\xi) \circ E(s;\hbar,x,\xi)
    =
    \frac{\der S}{\der s} e^{S(s;\hbar,x,\xi)/\hbar}.
\label{XE=dS/dsE}
\end{equation}
 By \lemref{lem:prod-exp-symbol} ($a \mapsto X$, $b \mapsto 1$, $p
 \mapsto 0$, $q \mapsto S$) and the technique in its proof, we can
 rewrite the left hand side as follows:
\begin{equation}
    X(\hbar,x,\xi) \circ E(s;\hbar,x,\xi)
    =
    Y(s;\hbar,x,x,\xi,\xi) e^{S(s;\hbar,x,\xi)/\hbar},
\label{XE=Y*eq}
\end{equation}
 where $Y(s;\hbar,x,y,\xi,\eta) = \sum_{k=0}^\infty Y_k$ and
 $Y_k(s;\hbar,x,y,\xi,\eta)$ are defined by
\begin{equation}
 \begin{split}
    &Y_{k+1}(s;\hbar,x,y,\xi,\eta)\\ &=
    \frac{1}{k+1} (\hbar \der_\xi \der_y Y_k(s;\hbar,x,y,\xi,\eta)
    + \der_\xi Y_k(s;\hbar,x,y,\xi,\eta) \der_y S(s;\hbar,y,\eta)),
\\
    &Y_0 (s;\hbar,x,y,\xi,\eta) = X(\hbar,x,\xi).
 \end{split}
\label{recursion:Y:1}
\end{equation}
 $Y_k(s)$ corresponds to $\psi_k$ in the proof of
 \lemref{lem:prod-exp-symbol}, while $w_k$ there corresponds to
 $\delta_{k,0} S(s)$.
 On the other hand, substituting \eqref{XE=Y*eq} into the left hand side
 of \eqref{XE=dS/dsE}, we have
\begin{equation}
    \frac{\der S}{\der s}(s;\hbar,x,\xi)
    =
    Y(s;\hbar,x,x,\xi,\xi).
\label{dS/ds=Y}
\end{equation}

 We rewrite the system \eqref{recursion:Y:1} and \eqref{dS/ds=Y} in
 terms of expansion of $S(s;\hbar,x,\xi)$ and
 $Y_k(s;\hbar,x,y,\xi,\eta)$ in powers of $s$ and $\hbar$:
\begin{equation}
 \begin{split}
    S(s;\hbar,x,\xi) &= \sum_{l=0}^\infty S^{(l)}(\hbar,x,\xi) s^l
    = \sum_{l=0}^\infty \sum_{n=0}^\infty S^{(l)}_n (x,\xi) \hbar^n s^l,
\\
    Y_k(s;\hbar,x,y,\xi,\eta) 
    &= \sum_{l=0}^\infty Y^{(l)}_k(\hbar,x,y,\xi,\eta) s^l
    = \sum_{l=0}^\infty \sum_{n=0}^\infty 
    Y^{(l)}_{k,n} (x,y,\xi,\eta) \hbar^n s^l,
 \end{split}
\label{expansion:S,Y}
\end{equation}
 The coefficient of $\hbar^n s^l$ in \eqref{recursion:Y:1} is
\begin{multline}
    Y^{(l)}_{k+1,n}(x,y,\xi,\eta)
\\    =
    \frac{1}{k+1}
    \left(
    \der_\xi \der_y Y^{(l)}_{k,n-1}(x,y,\xi,\eta)
    +
    \sum_{\substack{l'+l''=l \\ n'+n''=n}}
    \der_\xi Y^{(l')}_{k,n'}(x,y,\xi,\eta)
    \der_y S^{(l'')}_{n''}(y,\eta)
    \right),
\label{recursion:Y(l,k,n)}
\end{multline}
($Y^{(l)}_{k,-1}=0$) and
\begin{equation}
    Y^{(l)}_{0,n}(x,y,\xi,\eta) = \delta_{l,0} X_n(x,\xi),
\label{Yl0n}
\end{equation}
while \eqref{dS/ds=Y} gives
\begin{equation}
    S^{(l+1)}_n(x,\xi)
    = \frac{1}{l+1} \sum_{k=0}^\infty Y^{(l)}_{k,n}(x,x,\xi,\xi).
\label{recursion:S(l+1,n)}
\end{equation}
 We first show that these recursion relations consistently determine
 $Y^{(l)}_{k,n}$ and $S^{(l)}_n$. Then we prove that the infinite sum
 in \eqref{recursion:S(l+1,n)} is finite.

 Fix $n\geqq 0$ and assume that $Y^{(l)}_{k,0},\dots,Y^{(l)}_{k,n-1}$
 and $S^{(l)}_0,\dots,S^{(l)}_{n-1}$ have been determined for all
 $(l,k)$. (When $n=0$, $Y^{(l)}_{k,-1}=0$ as mentioned above and
 $S^{(l)}_{-1}$ can be ignored as it does not appear in the induction.)

\begin{enumerate}
 \item Since $E(s=0)=1$ by the definition \eqref{def:E(s)}, we have
       $S^{(0)}=0$. Hence
\begin{equation}
    S^{(0)}_n=0.
\label{S0n=0}
\end{equation}

 \item Note that
\begin{equation}
    \ordx Y^{(0)}_{0,n} \leqq -1
\label{ordY00n<=-1}
\end{equation}
      because of \eqref{Yl0n} and the assumption  $\ord X \leqq -1$. 

\item When $l=0$, the second sum in the right hand side of the recursion
      relation \eqref{recursion:Y(l,k,n)} is absent because of
      \eqref{S0n=0}. Hence if $n\geqq k+1$, we have
\[
    Y^{(0)}_{k+1,n} = \frac{1}{k+1}\der_\xi\der_y Y^{(0)}_{k,n-1}
    = \dots = \frac{1}{(k+1)!}(\der_\xi\der_y)^{k+1} Y^{(0)}_{0,n-k-1}
    = 0, 
\]
      since $Y^{(0)}_{0,n-k-1}$  does not depend on $y$ thanks to
      \eqref{Yl0n}. If $n<k+1$, the above expression becomes zero
      by $Y^{(0)}_{k-n+1,-1}=0$. Hence together with \eqref{Yl0n}, we
      obtain 
\begin{equation}
    Y^{(0)}_{k,n}=\delta_{k,0} X_n.
\label{Y0kn}
\end{equation}

 \item By \eqref{recursion:S(l+1,n)} we can determine $S^{(1)}_n$:
\begin{equation}
    S^{(1)}_n = \sum_{k=0}^\infty Y^{(0)}_{k,n} = Y^{(0)}_{0,n} = X_n.
\label{S(1,n)=Xn}
\end{equation}

 \item Fix $l_0\geqq 1$ and assume that for all $l=0,\dots,l_0-1$ and
       for all $k=0,1,2,\dots$, we have determined $Y^{(l)}_{k,n}$ and
       that for all $l=0,\dots,l_0$ we have determined
       $S^{(l)}_{n}$. (The steps (3) and (4) are for $l_0=1$.)

       Since $S^{(0)}_{n''}=0$ by \eqref{S0n=0}, the index $l'$ in the
       right hand side of the recursion relation
       \eqref{recursion:Y(l,k,n)} (with $l=l_0$) runs essentially from
       $0$ to $l_0-1$. Hence this relation determines
       $Y^{(l_0)}_{k+1,n}$ from known quantities for all $k\geqq 0$.

       Because of the initial condition $Y_0(s;x,y,\xi,\eta)=X(x,\xi)$
       (cf.\ \eqref{recursion:Y:1}) $Y_0$ does not depend on $s$, which
       means that its Taylor coefficients $Y^{(l_0)}_{0,n}$ vanish for
       all $l_0\geqq 1$:
\begin{equation}
    Y^{(l_0)}_{0,n} = 0.
\end{equation}
       Thus we have determined all
       $Y^{(l_0)}_{k,n}$ ($k=0,1,2,\dots$).
 
 \item We shall prove below that $Y^{(l_0+1)}_{k,n}=0$ if $k>l_0+n+1$.
       Hence the sum in \eqref{recursion:S(l+1,n)} is finite and
       $S^{(l_0+1)}_n$ can be determined. The induction proceeds 
       by incrementing $l_0$ by one.
\end{enumerate}

 Let us prove that $Y^{(l)}_{k,n}$'s determined above satisfy
\begin{gather}
    Y^{(l)}_{k,n}=0, \qquad\text{if\ }k>l+n,
\label{Ylkn=0}
\\
    \ordx Y^{(l)}_{k,n} \leqq -k-l-1.
\label{ordYlkn<-k-l-1}
\end{gather}
 (We define that $\ordx 0 = -\infty$.) In particular, the sum in
 \eqref{recursion:S(l+1,n)} is well-defined and
\begin{equation}
    \ordx S^{(l+1)}_n \leqq -l-1.
\label{ordS(l+1)<-l-1}
\end{equation} 

 If $n=-1$, both \eqref{Ylkn=0} and \eqref{ordYlkn<-k-l-1} are
 obvious. Fix $n_0\geqq 0$ and assume that we have proved \eqref{Ylkn=0}
 and \eqref{ordYlkn<-k-l-1} for $n<n_0$ and all $(l,k)$.
 
 When $n=n_0$ and $l=0$, \eqref{Ylkn=0} and \eqref{ordYlkn<-k-l-1} are
 true for all $k$ because of \eqref{Y0kn} and \eqref{ordY00n<=-1}. 

 Fix $l_0\geqq 0$ and assume that we have prove \eqref{Ylkn=0} and
 \eqref{ordYlkn<-k-l-1} for $n=n_0$, $l\leqq l_0$ and all $k$. As a
 result \eqref{ordS(l+1)<-l-1} is true for $l\leqq l_0$.

 For $k=0$, \eqref{Ylkn=0} is void and \eqref{ordYlkn<-k-l-1} is true
 because of \eqref{Yl0n} and $\ord X_n \leqq  -1$.

 Put $n=n_0$ and $l=l_0+1$ in \eqref{recursion:Y(l,k,n)} and assume that
 $k+1>(l_0+1)+n_0$. Then $k>(l_0+1)+(n_0-1)$, which guarantees that
 $Y^{(l_0+1)}_{k,n_0-1}=0$ by the induction hypothesis on $n$. As we
 mentioned in the step (5) above, $l'$ in the right hand side of
 \eqref{recursion:Y(l,k,n)} runs from $0$ to $l-1=l_0$. Hence, as we are
 assuming that $k>l_0+n$, we have $k>l'+n'$, which leads to
 $Y^{(l')}_{k,n'}=0$ by the induction hypothesis on $l$ and
 $n$. Therefore all terms in the right hand side of
 \eqref{recursion:Y(l,k,n)} vanish and we have
 $Y^{(l_0+1)}_{k+1,n_0}=0$, which implies \eqref{Ylkn=0} for $n=n_0$,
 $l=l_0+1$ and $k\geqq 1$.

 The estimate \eqref{ordYlkn<-k-l-1} is easy to check for $n=n_0$,
 $l=l_0+1$ and $k\geqq 1$ by the recursion relation
 \eqref{recursion:Y(l,k,n)}. (Recall once again that $\der_\xi$ lowers
 the order by one.)

 The step $l=l_0+1$ being proved, the induction proceeds with respect to
 $l$ and consequently with respect to $n$.

\bigskip
 In summary we have constructed $Y(s;\hbar,x,y,\xi,\eta)$ and
 $S(s;\hbar,x,\xi)$ satisfying \eqref{XE=Y*eq} and
 \eqref{dS/ds=Y}. Thanks to \eqref{ordS(l+1)<-l-1},
 $S_n(x,\xi)=\sum_{l=0}^\infty S^{(l)}_n(x,\xi)$ is meaningful as a
 power series of $\xi$. Thus \propref{prop:exp(:X:)=:exp(S):} is proved.
\end{proof}

\begin{proof}[Proof of \propref{prop::exp(S):=exp(:X:)}]
 We reverse the order of the previous proof. Namely, given
 $S(\hbar,x,\xi)$, we shall construct $X(\hbar,x,\xi)$ such that the
 corresponding $S(1;\hbar,x,\xi)$ in the above proof coincides with it.

 Suppose we have such $X(\hbar,x,\xi)$. Then the above procedure
 determine $Y^{(l)}_{k,n}$ and $S^{(l)}_n$. We expand them as follows:
\begin{align*}
    &S(\hbar,x,\xi) = \sum_{n=0}^\infty S_n(x,\xi) \hbar^n, &
    &S_n(x,\xi) = \sum_{j=1}^\infty S_{n,j}(x,\xi),
\\
    &X(\hbar,x,\xi) = \sum_{n=0}^\infty X_n(x,\xi) \hbar^n, &
    &X_n(x,\xi) = \sum_{j=1}^\infty X_{n,j}(x,\xi),
\\
    &S(s;\hbar,x,\xi)
    = \sum_{l=0}^\infty \sum_{n=0}^\infty S^{(l)}_n(x,\xi) \hbar^n s^l,
    &
    &S^{(l)}_n(x,\xi)=  \sum_{j=1}^\infty S^{(l)}_{n,j}(x,\xi),
\\
    &Y_k(s;\hbar,x,y,\xi,\eta)
    & &Y^{(l)}_{k,n}(x,y,\xi,\eta)
\\
    &= \sum_{l=0}^\infty  \sum_{n=0}^\infty 
    Y^{(l)}_{k,n}(x,y,\xi,\eta) \hbar^n s^l,
    & &= \sum_{j=1}^\infty 
    Y^{(l)}_{k,n,j}(x,y,\xi,\eta).
\end{align*}
 Here terms with index $j$ are homogeneous terms of degree $-j$ with
 respect to $\xi$ and $\eta$.

 At the end of this proof we shall determine $X_n$ by \eqref{Yl0n},
\begin{equation}
    X_n(x,\xi) = Y^{(0)}_{0,n}(x,y,\xi,\eta).
\label{Xn=Y(0,0,n)}
\end{equation}
 (In particular, $Y^{(0)}_{0,n}(x,y,\xi,\eta)$ should not depend on $y$
 and $\eta$.) For this purpose, $Y^{(0)}_{0,n}$ should be determined by
\begin{equation}
    Y^{(0)}_{0,n}(x,y,\xi,\eta)
    =
    S_n(x,\xi) 
    - \sum_{\substack{(l,k)\neq(0,0)\\l,k\geq 0}}
      \frac{1}{l+1} Y^{(l)}_{k,n}(x,x,\xi,\xi)
\label{Y00n=Sn-Ylkn}
\end{equation}
 because of \eqref{recursion:S(l+1,n)} and
 $S_n(x,\xi)=\sum_{l=1}^\infty S^{(l)}_n(x,\xi)$.

 Since $\ordx Y^{(l)}_{k,n}$ should be less than $-l-k$ (cf.\
 \eqref{ordYlkn<-k-l-1}), we expect
\begin{equation}
    Y^{(l)}_{k,n,1}=0
\label{Ylkn1=0}
\end{equation}
 for $(l,k)\neq(0,0)$. Hence picking up homogeneous terms of degree $-1$
 with respect to $\xi$ from \eqref{Y00n=Sn-Ylkn}, the following equation
 should hold: 
\begin{equation}
    Y^{(0)}_{0,n,1}=S_{n,1}
\label{Y00n1=Sn1}
\end{equation}
 All $Y^{(l)}_{k,n,1}$ are determined by the above two conditions,
 \eqref{Ylkn1=0} and \eqref{Y00n1=Sn1}. Note also that 
\begin{equation}
    Y^{(l)}_{0,n,j}=0 \ \text{for\ }l\neq 0
\label{Yl0nj=0}
\end{equation}
 because $Y_0$ should not depend on $s$ because of \eqref{Yl0n}.

 Having determined initial conditions in this way, we shall determine
 $Y^{(l)}_{k,n,j}$ inductively. To this end we rewrite the recursion
 relation \eqref{recursion:Y(l,k,n)} by \eqref{recursion:S(l+1,n)} and
 pick up homogeneous terms of degree $j$:
\begin{multline}
    Y^{(l)}_{k+1,n,j}(x,y,\xi,\eta)
    =
    \frac{1}{k+1}
    \Biggl(
    \der_\xi \der_y Y^{(l)}_{k,n-1,j-1}(x,y,\xi,\eta)+
\\
    +
    \sum_{\substack{l'+l''=l,l''\geq 1\\j'+j''=j-1, n'+n''=n\\ k''\geq 0}}
    \frac{1}{l''}
    \der_\xi Y^{(l')}_{k,n',j'}(x,y,\xi,\eta)
    \der_y Y^{(l''-1)}_{k'',n'',j''}(x,x,\xi,\xi)
    \Biggr)
\label{recursion:Y(l,k,n)<-Y}
\end{multline}
 (As before, terms like $Y^{(l)}_{k,-1,j-1}$ appearing the above
 equation for $n=0$ can be ignored.)

 Fix $n_0\geqq 0$ and assume that $Y^{(l)}_{k,0,j}, \dots,
 Y^{(l)}_{k,n_0-1,j}$ are determined for all $(l,k,j)$.
\begin{enumerate}
 \item First we determine $Y^{(l)}_{k,n_0,1}$ for all $(l,k)$ by
       \eqref{Y00n1=Sn1} and \eqref{Ylkn1=0}.

 \item Fix $j_0\geqq 2$ and assume that $Y^{(l)}_{k,n_0,j}$ are
       determined for $j=1,\dots,j_0-1$ and all $(l,k)$. (The above
       step is for $j_0=2$.)

       Since all the quantities in the right hand side of the recursion
       relation \eqref{recursion:Y(l,k,n)<-Y} with $j=j_0$ are known by
       the induction hypothesis, we can determine
       $Y^{(l)}_{k,n_0,j_0}$ for $l=0,1,2,\dots$ and $k=1,2,\dots$.

 \item Together with \eqref{Yl0nj=0}, $Y^{(l)}_{0,n_0,j_0}=0$
       for $l=1,2,\dots$, we have determined all $Y^{(l)}_{k,n_0,j_0}$
       except for the case $(l,k)=(0,0)$.

 \item It follows from \eqref{recursion:Y(l,k,n)<-Y}, \eqref{Y00n1=Sn1}
       and \eqref{Yl0nj=0} by induction that for all
       $Y^{(l)}_{k,n_0,j}$ determined in (1), (2) and (3),
\begin{equation}
    Y^{(l)}_{k,n,j}=0 \ \text{for\ }l+k+1>j.
\label{Ylknj=0:l+k+1>j}
\end{equation}
       This corresponds to $\ordx Y^{(l)}_{k,n_0}\leqq -l-k-1$
       \eqref{ordYlkn<-k-l-1} in the proof of
       \propref{prop:exp(:X:)=:exp(S):}. 

 \item We determine $Y^{(0)}_{0,n_0,j_0}$ by
\begin{equation}
    Y^{(0)}_{0,n_0,j_0}
    =
    S_{n_0,j_0}
    - \sum_{\substack{(l,k)\neq(0,0)\\l,k\geq 0}}
      \frac{1}{l+1} Y^{(l)}_{k,n_0,j_0}(x,x,\xi,\xi),
\label{Y00nj=Snj-Ylknj}
\end{equation}
       which is the homogeneous part of degree $-j_0$ in
       \eqref{Y00n=Sn-Ylkn}. The sum in the right hand side is finite
       thanks to \eqref{Ylknj=0:l+k+1>j}.

 \item The induction with respect to $j$ proceeds by incrementing $j_0$.
\end{enumerate}

 Thus all $Y^{(l)}_{k,n_0,j}$ are determined and $X_{n_0}$ is determined
 by $X_{n_0}(x,\xi)=\sum_{j=1}^\infty Y^{(0)}_{0,n_0,j}$ (cf.\
 \eqref{Xn=Y(0,0,n)}), which completes the proof of
 \propref{prop::exp(S):=exp(:X:)}.
\end{proof}

\section{Asymptotics of the tau function}
\label{sec:tau-function}

In this section we derive an $\hbar$-expansion 
\begin{equation}
    \log\tau(\hbar,t) = \sum_{n=0}^\infty \hbar^{n-2} F_n(t)
\label{tau}
\end{equation}
of the tau function (cf.\ \eqref{tau/tau}) from the $\hbar$-expansion
\eqref{S:h-expansion} of the $S$-function.  Note that we have suppressed
the variable $x$, which is understood to be absorbed in $t_1$.

Let us recall the fundamental relation \cite{djkm} between the wave
function \eqref{def:wave-func} and the tau function again:
\begin{equation}
    \Psi(t;z) =
    \frac{\tau(t-\hbar[z^{-1}])}{\tau(t)}
    e^{\hbar^{-1}\zeta(t,z)},
\end{equation}
where $[z^{-1}] =(1/z,1/2z^2, 1/3z^3,\dots)$ and
$\zeta(t,z)=\sum_{n=1}^\infty t_n z^n$. 
This implies that 
\begin{equation}
    \hbar^{-1} \hat S(t;z)
    =
    \left(
    e^{-\hbar D(z)} - 1
    \right) \log\tau(t)
\label{tau->w}
\end{equation}
where $\hat S(t;z) = S(t;z)- \zeta(t,z)$ and 
$D(z) =\sum_{j=1}^\infty \frac{z^{-j}}{j} \frac{\der}{\der t_j}$. 
Differentiating this with respect to $z$, we have
\begin{equation}
 \begin{split}
    \hbar^{-1} \frac{\der}{\der z} \hat S(t;z)
    =&
    - \hbar D'(z) e^{-\hbar D(z)} \log\tau(t)
\\
    =&
    - \hbar D'(z) (\hbar^{-1} \hat S(t;z) + \log\tau(t)),
 \end{split}
\label{(tau->w)'}
\end{equation}
where $D'(z):=\frac{\der}{\der z} D(z) = - \sum_{j=1}^\infty z^{-j-1}
\frac{\der}{\der t_j}$. Hence
\begin{equation}
    -\hbar D'(z) \log\tau(t)
    =
    \hbar^{-1} \left(\frac{\der}{\der z} + \hbar D'(z)\right)
    \hat S(t;z)
\label{w->tau}
\end{equation}
Multiplying $z^n$ to this equation and taking the residue, we obtain a
system of differential equations
\begin{equation}
    \hbar\frac{\der}{\der t_n} \log \tau(t)
    =
    \hbar^{-1} \Res_{z=\infty} z^n 
    \left(\frac{\der}{\der z} + \hbar D'(z)\right) \hat S(t;z)\, dz,
    \qquad
    n=1,2,\dots
\label{dtau/dtn}
\end{equation}
which is known to be integrable \cite{djkm}. 
We can thus determine the tau function $\tau(t)$, 
up to multiplication $\tau(t) \to c\tau(t)$ 
by a nonzero constant $c$, as a solution of \eqref{w->tau}.  

By substituting the $\hbar$-expansions
\begin{equation}
    \log\tau(t) = \sum_{n=0}^\infty \hbar^{n-2} F_n(t), \qquad
    \hat S(t;z) = \sum_{n=0}^\infty \hbar^n S_n(t;z),
\label{logtau,logw}
\end{equation}
into \eqref{w->tau}, we have
\begin{equation}
    \sum_{j=1}^\infty \sum_{n=0}^\infty z^{-j-1} \hbar^{n-1}
    \frac{\der F_n}{\der t_j}
    =
    \sum_{n=0}^\infty \left(
    \hbar^{n-1} \frac{\der S_n}{\der z}
    - 
    \sum_{j=1}^\infty z^{-j-1} \hbar^n \frac{\der S_n}{\der t_j}
    \right).
\label{S->F}
\end{equation}
Let us expand $S_n(t;z)$ into a power series of $z^{-1}$:
\begin{equation}
    S_n(t;z) = - \sum_{k=1}^\infty \frac{z^{-k}}{k} v_{n,k}.
\label{Sn:expand}
\end{equation}
(The notation is chosen to be consistent with our previous work, 
e.g., \cite{tak-tak:95}.) Comparing the coefficients of $z^{-j-1}
\hbar^{n-1}$ in \eqref{S->F}, we have the equations 
\begin{equation}
    \frac{\der F_n}{\der t_j}
    =
    v_{n,j} 
    + \sum_{\substack{k+l=j\\ k\geq 1, l\geq 1}}
      \frac{1}{l}\, \frac{\der v_{n-1,l}}{\der t_k}
    \qquad (v_{-1,j}=0), 
\label{dFn/dtj}
\end{equation}
which may be understood as defining equations of $F_n(t)$.  
According to what we have seen above, this system 
of differential equations is integrable 
and determines $F_n$ up to integration constants. 

\begin{rem}
\label{rem:genus-exp}
Tau functions in string theory and random matrices 
are known to have a {\em genus expansion} of the form
\begin{equation}
    \log \tau = \sum_{g=0} \hbar^{2g-2} \calF_g,
\label{genus-exp}
\end{equation}
where $\calF_g$ is the contribution from Riemann surfaces of genus $g$.
In contrast, general tau functions of the $\hbar$-dependent KP hierarchy
is not of this form, namely, odd powers of $\hbar$ can appear in the
$\hbar$-expansion of $\log\tau$. To exclude odd powers therein, we need
to impose conditions
\[
    0 =
    v_{2m+1,j} 
    + \sum_{\substack{k+l=j\\ k\geq 1, l\geq 1}}
      \frac{1}{l}\, \frac{\der v_{2m,l}}{\der t_k}
\]
 on $v_{n,j}$ or
\[
    0 =
    \frac{\der S_{2m+1}}{\der z}
    - 
    \sum_{j=1}^\infty z^{-j-1} \frac{\der S_{2m}}{\der t_j}
\]
 on $S_n$.
%
\end{rem}

\section{Concluding remarks}
\label{sec:conclusion}

We have presented a recursive construction of 
solutions of the $\hbar$-dependent KP hierarchy.  
The input of this construction is the pair $(f,g)$ 
of quantised canonical transformation. 
The main outputs are the dressing operator $W$ 
in the exponential form (\ref{W=exp(X)}), 
the wave function $\Psi$ in the WKB form (\ref{wave-func}) 
and the tau function with the quasi-classical expansion (\ref{tau}).  
Thus the $\hbar$-dependent KP hierarchy introduced 
in our previous work \cite{tak-tak:95} is no longer 
a heuristic framework for deriving the dispersionless 
KP hierarchy, but has its own raison d'\^{e}tre.  

A serious problem of our construction is that 
the recursion relations are extremely complicated.
In Appendix B, calculations are illustrated 
for the Kontsevich model \cite{kon:92}, \cite{dij:91},
\cite{adl-vaM:92}. As this example shows, 
this is by no means a practical way to construct 
a solution.  We believe that one cannot avoid 
this difficulty as far as general solutions 
are considered. 

Special solutions stemming from string theory and random matrices
\cite{mor:94}, \cite{dfgz} (e.g. the Kontsevich model) can admit a more
efficient approach such as the method of Eynard and Orantin
\cite{eyn-ora:07}.  Those methods are based on a quite different
principle.  In the method of Eynard and Orantin, it is the so called
``loop equation'' for correlation functions of random matrices.  The
loop equations amount to ``constraints'' on the tau function.  Eynard
and Orantin's ``topological recursion relations'' determine a solution
of those constraints rather than of an underlying integrable hierarchy;
it is somewhat surprising that a solution of those constraints gives a
tau function.

Lastly, let us mention that the results of this paper 
can be extended to the Toda hierarchy.  
That case will be treated in a forthcoming paper.

\appendix
\section{Proof of formulae \eqref{W=exp(Xi)exp(Xi-1)} and
 \eqref{tildeXi->Xi}}
\label{app:proof-campbell-hausdorff}

In this appendix we prove the factorisation of $W$
\eqref{W=exp(Xi)exp(Xi-1)} and an auxiliary formula
\eqref{tildeXi->Xi}. 

The main tool in this appendix is the Campbell-Hausdorff theorem: 
\begin{equation}
    \exp(X) \exp(Y)
    =
    \exp\left(
     \sum_{n=0}^\infty c_n(X, Y)
    \right),
\label{CH}
\end{equation}
where $c_n(X,Y)$ is determined recursively: 
\begin{equation}
 \begin{split}
    &c_1(X,Y) = X+Y,
\\
    &c_{n+1}(X,Y)
    =
    \frac{1}{n+1} \Biggl( 
    \frac{1}{2} [X-Y, c_n] +
\\
    &+
    \sum_{p\geq 1, 2p \leq n}
     K_{2p} \sum_{\substack{(k_1,\dots,k_{2p})\\ k_1+\cdots+k_{2p}=n}}
     [c_{k_1},[\cdots, [ c_{k_{2p}}, X+Y]\cdots]]
    \Biggr).
 \end{split}
\label{def:cn}
\end{equation}
The coefficients $K_{2p}$ are defined by \eqref{def:K2p}. See, for
example, \cite{bourbaki}.

First we prove
\begin{multline}
    \exp( \hbar^{-1} X(x,t,\hbar\der) )
\\
    =
    \exp\left(
     \hbar^{i-1} \tilde X'_i + 
     (\text{terms of $\hbar$-order $<-i+1$})
    \right)    
    \exp\left( \hbar^{-1} X^{(i-1)} \right),
\label{exp(X)=exp(Xi)exp(Xi-1)}
\end{multline}
where the principal symbol of $\tilde X'_i$ is
\begin{equation}
    \symh(\tilde X'_i) := 
    \sum_{n=1}^\infty \frac{(\ad_{\{,\}} \symh(X_0) )^{n-1}}{n!}
    \symh(X_i),
\label{def:tildeX'i}
\end{equation}
as is defined in \eqref{tildeXi->Xi}. For simplicity, let us denote
\begin{equation}
    A := \frac{1}{\hbar} X^{(i-1)}
       = \frac{1}{\hbar} \sum_{j=0}^{i-1} \hbar^j X_j, \qquad
    B := \frac{1}{\hbar} \sum_{j=i}^\infty \hbar^j X_j.
\label{def:A,B}
\end{equation}
Note that $A+B=X/\hbar$ and $\ordh A\leqq 1$, $\ordh B\leqq -i+1$. We
prove the following by induction:
\begin{equation}
    C_n := c_n(A+B,-A) =
    \frac{(\ad A)^{n-1}}{n!} (B)
    +
    (\text{terms of $\hbar$-order $<-i+1$}).
\label{cn(A+B,-A)}
\end{equation}
This is obvious for $n=1$ since $C_1=(A+B)+(-A)=B$. Assume that
\eqref{cn(A+B,-A)} is true for $n=1,\dots,N$. This means, in particular,
$\ordh C_n \leqq \ordh B \leqq 0$ ($1\leqq n \leqq N$), which implies
that for any operator $Z$ $\ordh [C_n,Z]$ is less than $\ordh Z$ by more
than one. Hence the term of the highest $\hbar$-order in the recursive
definition \eqref{def:cn} with $X=A+B$, $Y=-A$ is the first term. More
precisely, it is decomposed as
\begin{equation*}
    \frac{1}{N+1} \cdot \frac{1}{2} [(A+B)-(-A), C_N]
    =
    \frac{1}{N+1}[A,C_N] + \frac{1}{2(N+1)}[B,C_N],
\end{equation*}
and the first term in the right hand side has the highest
$\hbar$-order. By the induction hypothesis and $\ordh A \leqq 1$, we
have 
\begin{equation}
 \begin{split}
    \frac{1}{N+1}[A, C_N] &=
    \frac{1}{N+1}\left[ A,
    \frac{(\ad A)^{N-1}}{N!} (B)
    +
    (\text{terms of $\hbar$ order $<-i+1$})
    \right]
\\
    &=
    \frac{(\ad A)^N}{(N+1)!} (B)
    +
    (\text{terms of $\hbar$-order $<-i+1$}).
 \end{split}
\end{equation}
This proves \eqref{cn(A+B,-A)} for all $n$. Taking its symbol of order
$-i+1$, we have
\begin{equation}
    \symh( c_n(A+B,-A)) =
    \frac{(\ad_{\{,\}} \symh(A))^{n-1}}{n!} \symh(B),
\label{sigma(cn(A+B,-A))}
\end{equation}
which gives the terms of \eqref{def:tildeX'i}. Substituting this into
the Campbell-Hausdorff formula \eqref{CH}, we have
\eqref{exp(X)=exp(Xi)exp(Xi-1)}. 

By factorisation \eqref{exp(X)=exp(Xi)exp(Xi-1)}, we can factorise
$W=\exp(X/\hbar) (\hbar\der)^\alpha$ as follows
($\alpha^{(i-1)}:=\sum_{j=0}^{i-1}\hbar^j \alpha_j$):
\begin{equation}
 \begin{split}
   &\exp( \hbar^{-1} X(x,t,\hbar\der) ) (\hbar\der)^\alpha
\\
    =&
    \exp\left(
     \hbar^{i-1} \tilde X'_i + 
     (\text{terms of $\hbar$-order $<-i+1$})
    \right)    
    \exp\left( \hbar^{-1} X^{(i-1)} \right) \times
\\
    &\times
    \exp\left(
     \hbar^{i-1} \alpha_i \log(\hbar\der) + 
     (\text{terms of $\hbar$-order $<-i+1$})
    \right) \times
\\
    &\times
    \exp\left( \hbar^{-1} \alpha^{(i-1)} \log(\hbar\der)\right)
\\
    =&
    \exp\left(
     \hbar^{i-1} \tilde X'_i + 
     (\text{terms of $\hbar$-order $<-i+1$})
    \right)    
    \times
\\
    &\times
    \exp\left(
    e^{\ad (\hbar^{-1} X^{(i-1)})}
     \bigl(
       \hbar^{i-1} \alpha_i \log(\hbar\der) + 
      (\text{terms of $\hbar$-order $<-i+1$})
     \bigr)
    \right) \times
\\
    &\times
    \exp\left( \hbar^{-1} X^{(i-1)} \right)
    \exp\left( \hbar^{-1} \alpha^{(i-1)} \log(\hbar\der)\right).
 \end{split}
\label{W=exp(Xi)exp(Xi-1):temp}
\end{equation}
Since the symbol of order $-i+1$ of 
$
    e^{\ad (\hbar^{-1} X^{(i-1)})}
     \bigl(
       \hbar^{i-1} \alpha_i \log(\hbar\der)
     \bigr)
$
is $e^{\ad_{\{,\}} \symh(X_0)}(\alpha_i \log\xi)$,
\eqref{W=exp(Xi)exp(Xi-1):temp} is rewritten as
\eqref{W=exp(Xi)exp(Xi-1)} by using the Campbell-Hausdorff formula
\eqref{CH} once again.

\bigskip
In order to recover $X_i$ from $\tilde X_i$ (or $\tilde X'_i$), we have
only to invert the definition \eqref{def:tildeX'i}. In the definition
\eqref{def:tildeX'i} of the map $X_i\mapsto \tilde X_i'$ we substitute
$\ad_{\{,\}}(\symh(X_0))$ in the equation
\begin{equation*}
    \frac{e^t-1}{t} = \sum_{n=1}^\infty \frac{t^{n-1}}{n!}.
\end{equation*}
Hence substitution $t=\ad_{\{,\}}(\symh(X_0))$ in its inverse
\begin{equation*}
    \frac{t}{e^t-1} 
    = 1 - \frac{t}{2} + \sum_{p=1}^\infty K_{2p} t^{2p}
\end{equation*}
gives the inverse map $\tilde X'_i \mapsto X_i$. Here the coefficients
$K_{2p}$ are defined in \eqref{def:K2p}. Hence equation
\eqref{tildeXi->Xi}:
\begin{equation*}
    \symh(X_i) 
    = \symh(\tilde X'_i) 
    - \frac{1}{2} \{\symh(X_0),\symh(\tilde X'_i)\}
    + \sum_{p=1}^\infty K_{2p}
      (\ad_{\{,\}} (\symh(X_0)))^{2p} \symh(\tilde X'_i)
\end{equation*}
gives the symbol of $X_i$.
%

\section{Example (Kontsevich model)}
\label{sec:kontsevich-model}

According to Adler and van Moerbeke \cite{adl-vaM:92}, the solution of
the KP hierarchy arising in the Witten-Kontsevich theorem \cite{kon:92}
satisfies
\begin{equation}
    (L^2)_{\leq -1} = 0, \qquad
    \left(\frac12 ML^{-1} -\frac14 \hbar L^{-2} - L \right)_{\leq -1}
     = 0.
\label{konts:condition}
\end{equation}
This corresponds to the case where 
\begin{equation}
    f(\hbar,x,\hbar\der) = (\hbar\der)^2, \qquad
    g(\hbar,x,\hbar\der) = 
    \frac{1}{2} x(\hbar\der)^{-1} - \frac{\hbar}{4} (\hbar\der)^{-2}
    + \hbar\der.
\label{konts:f,g}
\end{equation}

Let us apply our procedure to this case. We fix all the time variables
$t_n$ to $0$, which means that we restrict ourselves to the so-called
``small phase space'' in topological string theory
\cite{dij:91}. (The first time variable $t_1$ can be re-introduced by
shifting $x$ to $t_1+x$.)

\bigskip
To begin with, let us determine the leading terms of $X$ and $\alpha$,
which are the initial data for our procedure.
The dispersionless limit of $(L,M)$ satisfies
\begin{equation}
    (\calL^2)_{\leq -1}=0, \qquad
    \left(\frac{1}{2} \calM\calL^{-1} - \calL\right)_{\leq -1} =0.
\label{konts:displess-cond}
\end{equation}
Since $\calL$ has the form \eqref{calL}, $\calL^2$ should be a second
order polynomial of $\xi$: $\calL^2=\xi^2 + u(x)$. When $t_n=0$, $\calM$
has the form
\[
    \calM = 
    x + \alpha_0 \calL^{-1} + \sum_{n=1}^\infty v_{0,n}\calL^{-n-1}.
\]
(See \eqref{calM}.) Using this and $\calL=\xi(1+u(x)\xi^{-2})^{1/2}$, we
have
\[
    \frac{1}{2} \calM\calL^{-1} - \calL
    =
    - \xi + \left(\frac{x}{2}-\frac{u}{2}\right)\xi^{-1}
    + \frac{\alpha_0}{2} \xi^{-2} + O(\xi^{-3}).
\]
Hence, due to the second equation of \eqref{konts:displess-cond}, we
have $u(x)=x$, $\alpha_0=0$ and, consequently,
\begin{equation}
    \calP_0 := \calL^2 = \xi^2 + x, \qquad
    \calQ_0 := \frac{1}{2} \calM\calL^{-1} - \calL
             = - \xi.
\label{konts:P0,Q0}
\end{equation}
Combining them, we have the following expression for $\calM$:
\begin{equation}
 \begin{split}
    \calM =&
    2\calL(\calL - \xi)
    = 2(\xi^2 + x) - 2\xi^2 \bigl(1 + x \xi^{-2}\bigr)^{1/2}
\\
    =&
    x - \sum_{n=2}^\infty 2 \binom{1/2}{n} x^n \xi^{-2n+2}
\\
    =&
    x
    +\frac{x^2}{4 \xi^2}
    -\frac{x^3}{8 \xi^4}
    +\frac{5 x^4}{64 \xi^6}
    -\frac{7 x^5}{128 \xi^8}
    +\frac{21 x^6}{512 \xi^{10}}
    + O(\xi^{-12})
 \end{split}
\label{konts:calM}
\end{equation}

On the other hand, from the expression $\calM = \exp(\ad_{\{,\}} X_0)x$
follows 
\begin{equation}
 \begin{split}
    \calM =&
    \sum_{n=0}^\infty \frac{\ad_{\{,\}} X_0}{n!} x
\\
    =&
    x - \chi_{0,1} \xi^{-2} - 2 \chi_{0,2} \xi^{-3}
\\
    &+
    \left( -3\chi_{0,3}-\frac{\chi_{0,1}\chi'_{0,1}}{2} \right)\xi^{-4} 
    +
    ( -4\chi_{0,4}-2\chi_{0,2}\chi'_{0,1}) \xi^{-5} + \cdots,
 \end{split}
\label{konts:calM=exp(ad)x}
\end{equation}
where $\chi_{0,k}$ are the coefficients in the expansion of $X_0$
\eqref{def:Xi}. Comparing \eqref{konts:calM} and
\eqref{konts:calM=exp(ad)x}, we can determine $\chi_{0,k}$ inductively
and hence $X_0$ is determined:
\begin{equation*}
 \begin{split}
    X_0=&
    -\frac{x^2}{4} (\hbar\der)^{-1}
    +\frac{x^3}{48} (\hbar\der)^{-3}
    -\frac{x^4}{384}(\hbar\der)^{-5}
\\
    &
    +\frac{x^6}{6144}(\hbar\der)^{-9}
    -\frac{x^7}{61440}(\hbar\der)^{-11}
    +\cdots.
 \end{split}
\end{equation*}

\bigskip
Having determined $X_0$ and $\alpha_0$, we can start the algorithm
discussed in \secref{sec:recursion}. In the step 1 for $i=1$ we define
$P^{(0)}$ and $Q^{(0)}$ by \eqref{def:Pi-1} and \eqref{def:Qi-1}:
\begin{align*}
    P^{(0)} &=
    (\hbar\der)^2+x
    +\frac{\hbar}{2}(\hbar\der)^{-1}
    -\frac{\hbar x}{4}(\hbar\der)^{-3}
    +\frac{\hbar^2}{8}(\hbar\der)^{-4}
    +\frac{5 \hbar x^2}{32}(\hbar\der)^{-5}
\\
    &
    -\frac{41 \hbar^2 x}{96}(\hbar\der)^{-6}
    +\left(\frac{23 \hbar^3}{96}
          -\frac{3 \hbar x^3}{32}\right)(\hbar\der)^{-7}
    +\frac{301 \hbar^2 x^2}{384}(\hbar\der)^{-8}
\\
    &
    +\left(-\frac{83 \hbar^3 x}{48}
           +\frac{53 \hbar x^4}{1024}\right)(\hbar\der)^{-9}
    +\left( \frac{191 \hbar^4}{192}
           -\frac{543 \hbar^2 x^3}{512}\right)(\hbar\der)^{-10}
\\
    &
    +\left( \frac{8783 \hbar^3 x^2}{1536}
           -\frac{119 \hbar x^5}{4096}\right)(\hbar\der)^{-11}
\commentout{
    +\left(-\frac{4375 \hbar^4 x}{384}
           +\frac{28771 \hbar^2 x^4}{24576}\right)(\hbar\der)^{-12}
\\
    &
    +\left(\frac{1271 \hbar^5}{192}-\frac{9763 \hbar^3 x^3}{768}
          +\frac{2215 \hbar x^6}{98304}\right)(\hbar\der)^{-13}
\\
    &
    +\left(\frac{32429 \hbar^4 x^2}{576}
          -\frac{37541 \hbar^2 x^5}{32768}\right) (\hbar\der)^{-14}
\\
    &
    +\left(-\frac{972743 \hbar^5 x}{9216}
           +\frac{2136313 \hbar^3 x^4}{98304}
           -\frac{2377 \hbar x^7}{98304}\right)(\hbar\der)^{-15}
\\
    &
    +\left(\frac{283867 \hbar^6}{4608}
          -\frac{16476523 \hbar^4 x^3}{92160}
          +\frac{1348327 \hbar^2 x^6}{1179648}\right)(\hbar\der)^{-16}
\\
    &
    +\left(\frac{8670875 \hbar^5 x^2}{12288}
          -\frac{3053023 \hbar^3 x^5}{98304}
          +\frac{74129 \hbar x^8}{3145728}\right)(\hbar\der)^{-17}
}
    + \cdots,
\\
    Q^{(0)} &= -(\hbar\der)
    -\frac{\hbar}{4} (\hbar\der)^{-2}
    +\frac{3 \hbar x}{8} (\hbar\der)^{-4}
    -\frac{3 \hbar^2}{8} (\hbar\der)^{-5}
    -\frac{29 \hbar x^2}{64} (\hbar\der)^{-6}
\\
    &
    +\frac{157 \hbar^2 x}{96} (\hbar\der)^{-7}
    +\left(-\frac{49 \hbar^3}{32}+\frac{\hbar x^3}{2}\right)(\hbar\der)^{-8}
    -\frac{519 \hbar^2 x^2}{128} (\hbar\der)^{-9}
\\
    &
    +\left(\frac{4345 \hbar^3 x}{384}
          -\frac{1077 \hbar x^4}{2048}\right)(\hbar\der)^{-10}
    +\left(-\frac{1339 \hbar^4}{128}
           +\frac{3961 \hbar^2 x^3}{512}\right)(\hbar\der)^{-11}
\commentout{
\\
    &
    +\left(-\frac{5571}{128} \hbar^3 x^2
           +\frac{4427 \hbar x^5}{8192}\right)(\hbar\der)^{-12}
    +\left(\frac{14021 \hbar^4 x}{128}
           -\frac{38917 \hbar^2 x^4}{3072}\right)(\hbar\der)^{-13}
\\
    &
    +\left(-\frac{12939 \hbar^5}{128}
           +\frac{123345 \hbar^3 x^3}{1024}
           -\frac{108463 \hbar x^6}{196608}\right)(\hbar\der)^{-14}
\\
    &
    +\left(-\frac{10626799 \hbar^4 x^2}{18432}
           +\frac{1849331 \hbar^2 x^5}{98304}\right)(\hbar\der)^{-15}
\\
    &
    +\left(\frac{12647275 \hbar^5 x}{9216}
          -\frac{26731385 \hbar^3 x^4}{98304}
          +\frac{13855 \hbar x^7}{24576}\right)(\hbar\der)^{-16}
\\
    &
    +\left(-\frac{1298533 \hbar^6}{1024}
           +\frac{129676111 \hbar^4 x^3}{61440}
           -\frac{30900365 \hbar^2 x^6}{1179648}\right)(\hbar\der)^{-17}
\\
    &
    +\left(-\frac{3394224569 \hbar^5 x^2}{368640}
           +\frac{210212819 \hbar^3 x^5}{393216}
           -\frac{1207003 \hbar x^8}{2097152}\right)(\hbar\der)^{-18}
}
\\
    &+ \cdots.
\end{align*}
We extract terms (symbols) of $\hbar$-order $0$ from the
$\hbar$-expansion of them:
\begin{equation*}
    \calP_0(x,\xi)=\xi^2+x, \qquad
    \calQ_0(x,\xi)=-\xi,
\end{equation*}
and those of $\hbar$-order $-1$:
\begin{align*}
    \calP^{(0)}_1(x,\xi)=&
    \frac{1}{2 }\xi^{-1}
    -\frac{x}{4 }\xi^{-3}
    +\frac{5 x^2}{32 }\xi^{-5}
    -\frac{3 x^3}{32 }\xi^{-7}
    +\frac{53 x^4}{1024 }\xi^{-9}
\\  &
    -\frac{119 x^5}{4096 }\xi^{-11}
\commentout{
    +\frac{2215 x^6}{98304 }\xi^{-13}
    -\frac{2377 x^7}{98304 }\xi^{-15}
    +\frac{74129 x^8}{3145728 } \xi^{-17}
}
    + \cdots
\\
    \calQ^{(0)}_1(x,\xi)=&
    -\frac{1}{4} \xi^{-2}
    +\frac{3 x}{8} \xi^{-4}
    -\frac{29x^2}{64} \xi^{-6}
    +\frac{x^3}{2} \xi^{-8}
    -\frac{1077 x^4}{2048} \xi^{-10}
\commentout{
    +\frac{4427 x^5}{8192} \xi^{-12}
\\  &
    -\frac{108463 x^6}{196608} \xi^{-14}
    +\frac{13855 x^7}{24576} \xi^{-16}
    -\frac{1207003 x^8}{2097152} \xi^{-18}
}
    + \cdots.
\end{align*}
Then, following \eqref{ai+tildeXi=int}, we determine $\alpha_0$ and
$\calX_0$ by
\begin{equation*}
 \begin{split}
    \alpha_1 \log\xi + \tilde\calX_1:=&
    \int\left(
    \dfrac{\der \calQ_0}{\der \xi} \calP^{(0)}_1
    -
    \dfrac{\der \calP_0}{\der \xi} \calQ^{(0)}_1
    \right)_{\leq -1}
    d\xi
\\
    =&
    \frac{x}{4} \xi^{-2}
    -\frac{3 x^2}{16} \xi^{-4}
    +\frac{29 x^3}{192} \xi^{-6}
    -\frac{x^4}{8} \xi^{-8}
    +\frac{1077 x^5}{10240} \xi^{-10}
\commentout{
\\  &
    -\frac{4427 x^6}{49152} \xi^{-12}
    +\frac{108463 x^7}{1376256} \xi^{-14}
    -\frac{13855 x^8}{196608} \xi^{-16}
}
    + \cdots.
 \end{split}
\end{equation*}
Since $\log$ term is absent, $\alpha_1=0$ and the above expression is
$\tilde\calX_1$ itself, which is also equal to $\tilde\calX'_1$ defined
by \eqref{tildeXi->Xi:symbol}. Then we can compute $X_1$ by the formulae
\eqref{tildeXi->Xi:symbol} and \eqref{def:Xi}:
\begin{equation*}
 \begin{split}
    X_1=&
    \frac{x}{4} (\hbar\der)^{-2}
    -\frac{3 x^2}{32} (\hbar\der)^{-4}
    +\frac{5 x^3}{192} (\hbar\der)^{-6}
    -\frac{9 x^5}{2048} (\hbar\der)^{-10}
    + \cdots
\end{split}
\end{equation*}
We can repeat the procedure Step 1, 2, 3 for $i=2$ again. The results
are
\begin{equation*}
 \begin{split}
    P^{(1)}
    =&
    (\hbar\der)^2+x 
    +\frac{3 \hbar^2}{16} (\hbar\der)^{-4}
    -\frac{23 \hbar^2 x}{96 } (\hbar\der)^{-6}
    +\frac{9 \hbar^3}{16 } (\hbar\der)^{-7}
    +\frac{19 \hbar^2 x^2}{384 } (\hbar\der)^{-8}
\\  &
    -\frac{707 \hbar^3 x}{384 } (\hbar\der)^{-9}
    +\left(\frac{861 \hbar^4}{256}
          +\frac{155 \hbar^2 x^3}{512}\right) (\hbar\der)^{-10}
    +\frac{2145 \hbar^3 x^2}{1024 }(\hbar\der)^{-11}
\commentout{
\\  &
    +\left(-\frac{607 \hbar^4 x}{32}
           -\frac{13805 \hbar^2 x^4}{24576}\right) (\hbar\der)^{-12}
    +\left(\frac{15207 \hbar^5}{512}
          +\frac{969 \hbar^3 x^3}{512}\right) (\hbar\der)^{-13}
\\  &
    +\left(\frac{1554395 \hbar^4 x^2}{36864}
          +\frac{15527 \hbar^2 x^5}{32768}\right) (\hbar\der)^{-14}
\\  &
    +\left(-\frac{4521953 \hbar^5 x}{18432}
           -\frac{44949 \hbar^3 x^4}{4096}\right) (\hbar\der)^{-15}
\\  &
    +\left(\frac{89095 \hbar^6}{256}
          -\frac{19150591 \hbar^4 x^3}{1474560}
          -\frac{117707 \hbar^2 x^6}{1179648}\right) (\hbar\der)^{-16}
\\  &
    +\left(\frac{136733991 \hbar^5 x^2}{163840}
      +\frac{7393931 \hbar^3 x^5}{393216}
      +\frac{1207003 \hbar x^8}{1048576}\right) (\hbar\der)^{-17}
}
\\
    &
    + \cdots,
\\
    Q^{(1)}=&
    - \hbar\der
    -\frac{\hbar^2}{4 } (\hbar\der)^{-5}
    +\frac{143 \hbar^2 x}{192 } (\hbar\der)^{-7}
    -\frac{611 \hbar^3}{384 } (\hbar\der)^{-8}
    -\frac{85 \hbar^2 x^2}{64 } (\hbar\der)^{-9}
\\  &
    +\frac{2205 \hbar^3 x}{256} (\hbar\der)^{-10}
    +\left(-\frac{1795 \hbar^4}{128}
           +\frac{1885 \hbar^2 x^3}{1024}\right) (\hbar\der)^{-11}
\commentout{
\\  &
    -\frac{12975 \hbar^3 x^2}{512} (\hbar\der)^{-12}
    +\left(\frac{86081 \hbar^4 x}{768}
          -\frac{110605 \hbar^2 x^4}{49152}\right) (\hbar\der)^{-13}
\\  &
    +\left(-\frac{498661 \hbar^5}{3072}
           +\frac{1332949 \hbar^3 x^3}{24576}\right) (\hbar\der)^{-14}
\\  &
    +\left(-\frac{68882177 \hbar^4 x^2}{147456}
           +\frac{528565 \hbar^2 x^5}{196608}\right) (\hbar\der)^{-15}
\\  &
    +\left(\frac{28350877 \hbar^5 x}{16384}
          -\frac{12602615 \hbar^3 x^4}{131072}\right) (\hbar\der)^{-16}
\\  &
    +\left(-\frac{43208869 \hbar^6}{18432}
           +\frac{336110953 \hbar^4 x^3}{245760}
           -\frac{488173 \hbar^2 x^6}{147456}\right) (\hbar\der)^{-17}
\\  &
    +\left(-\frac{3083193377 \hbar^5 x^2}{327680}
           +\frac{121401941 \hbar^3 x^5}{786432}
           -\frac{1207003 \hbar x^8}{2097152}\right) (\hbar\der)^{-18}
}
    + \cdots.
 \end{split}
\end{equation*}
Collecting the terms of $\hbar$-order $-2$, we have
\begin{align*}
    \calP^{(1)}_2 =&
    \frac{3}{16 }\xi^{-4}
    -\frac{23 x}{96} \xi^{-6}
    +\frac{19 x^2}{384} \xi^{-8}
    +\frac{155 x^3}{512} \xi^{-10}
    + \cdots,
\\
    \calQ^{(1)}_2 =&
    -\frac{1}{4} \xi^{-5}
    +\frac{143 x}{192} \xi^{-7}
    -\frac{85 x^2}{64} \xi^{-9}
    +\frac{1885 x^3}{1024} \xi^{-11}
    + \cdots.
\end{align*}
From this result follows
\begin{equation*}
 \begin{split}
    \alpha_2 \log\xi + \tilde\calX_2:=&
    \int\left(
    \dfrac{\der \calQ_0}{\der \xi} \calP^{(1)}_2
    -
    \dfrac{\der \calP_0}{\der \xi} \calQ^{(1)}_2
    \right)_{\leq -1}
    d\xi
\\
    =&
    -\frac{5}{48} \xi^{-3}
    +\frac{x}{4} \xi^{-5}
    -\frac{143 x^2}{384} \xi^{-7}
    +\frac{85 x^3}{192} \xi^{-9}
    -\frac{1885 x^4}{4096} \xi^{-11}
    +\cdots.
 \end{split}
\end{equation*}
This means $\alpha_2=0$ and $\tilde \calX'_2$ is equal to $\tilde
\calX_2$, which is equal to the above expression. Substituting these
results in \eqref{tildeXi->Xi:symbol}, we have
\begin{equation*}
 \begin{split}
    X_2 =&
    -\frac{5}{48} (\hbar\der)^{-3}
    +\frac{11 x}{64} (\hbar\der)^{-5}
    -\frac{85 x^2}{768 } (\hbar\der)^{-7}
    +\frac{435 x^4}{8192 } (\hbar\der)^{-11}
    + \cdots.
 \end{split}
\end{equation*}
Substituting this into \eqref{def:Pi-1} and \eqref{def:Qi-1} for $i=3$,
we have 
\begin{equation*}
 \begin{split}
    P^{(2)}=&
    (\hbar\der)^2+x
    +\frac{5 \hbar^3}{48} (\hbar\der)^{-7}
    +\frac{425 \hbar^3 x}{768 }(\hbar\der)^{-9}
    +\frac{3205 \hbar^4}{4096 }(\hbar\der)^{-10}
\\  &
    -\frac{5865 \hbar^3 x^2}{2048 }(\hbar\der)^{-11}
    + \cdots,
\\
    Q^{(2)}=&
    -(\hbar\der)
    -\frac{155 \hbar^3}{384} (\hbar\der)^{-8}
    +\frac{685 \hbar^3 x}{512} (\hbar\der)^{-10}
    -\frac{41395 \hbar^4}{8192} (\hbar\der)^{-11}
    + \cdots.
 \end{split}
\end{equation*}
By extracting the terms of $\hbar$-order $-3$, we have
\begin{equation}
 \begin{split}
    \calP^{(2)}_3=&
    \frac{5}{48} \xi^{-7}
    +\frac{425 x}{768} \xi^{-9}
    -\frac{5865 x^2}{2048} \xi^{-11}
    + \cdots,
\\
    Q^{(2)}_3=&
    -\frac{155}{384} \xi^{-8}
    +\frac{685 x}{512} \xi^{-10}
    + \cdots.
 \end{split}
\end{equation}
Hence,
\begin{equation*}
 \begin{split}
    \alpha_3 \log\xi + \tilde\calX_3:=&
    \int\left(
    \dfrac{\der \calQ_0}{\der \xi} \calP^{(2)}_3
    -
    \dfrac{\der \calP_0}{\der \xi} \calQ^{(2)}_3
    \right)_{\leq -1}
    d\xi
\\
    =&
    -\frac{15}{128} \xi^{-6}
    +\frac{155 x}{384} \xi^{-8}
    -\frac{685 x^2}{1024} \xi^{-10}
    + \cdots,
 \end{split}
\end{equation*}
which implies $\alpha_3 = 0$ and $\tilde X_3 = \tilde X'_3$ is equal to
the above expression. Thus, again by \eqref{tildeXi->Xi:symbol}, we have
\begin{equation*}
 \begin{split}
    X_3 =&
    -\frac{15}{128} (\hbar\der)^{-6}
    +\frac{175 x}{768} (\hbar\der)^{-8}
    + \cdots
 \end{split}
\end{equation*}
Consequently, application of 
$
    \exp(\Ad((X_0+\hbar X_1 + \hbar^2 X_2 + \hbar^3 X_3)/\hbar))
$
to $f$ and $g$ gives
\begin{equation*}
 \begin{split}
    P^{(3)} =&
    (\hbar\der)^2+x
    -\frac{3395 \hbar^4}{4096} (\hbar\der)^{-10}
    + \cdots
\\
    Q^{(3)} =& 
    -(\hbar\der)
    -\frac{3395 \hbar^4}{8192}(\hbar\der)^{-11}
    + \cdots,
 \end{split}
\end{equation*}
which are differential operators up to $\hbar$-order $-3$.

\end{document}